\documentclass[a4paper,10pt]{article}

\newcommand{\typeof}{0}            %

\usepackage{ifthen}
\newcommand{\condinc}[2]{\ifthenelse{\equal{\typeof}{0}}{#1}{#2}}

\condinc{}
{
\usepackage{microtype}
\newcommand{\rtacat}[1]{\vskip\baselineskip\noindent \textcolor{darkgray}{\fontsize{9.5}{12.5}\sffamily\bfseries Category\enskip} #1}
}

\usepackage{etex}
\usepackage{graphicx}
\usepackage{verbatim}
\usepackage{stmaryrd}
\usepackage{amsmath}
\usepackage{amsfonts,amssymb}
\usepackage{bussproofs}
\usepackage{fancybox}
\usepackage{ifthen}
\usepackage{color}
\condinc{\usepackage{a4wide}}{}
\usepackage{hyperref}
\hypersetup{linktocpage}

\condinc{
  \newtheorem{definition}{Definition}
  \newtheorem{theorem}{Theorem}
  \newtheorem{lemma}[theorem]{Lemma}
  
  \newtheorem{remark}{Remark}
  \newtheorem{corollary}[theorem]{Corollary}
  \newtheorem{proposition}[theorem]{Proposition}
\newenvironment{proof}{\begin{trivlist}
      \item[\hskip \labelsep {\bfseries Proof.}]}{\hfill $\Box$ \end{trivlist}}
}{}

\newcommand{\mellies}{Melli{\`e}s}

\newcommand{\deff}[1]{\emph{#1}}

\newcommand{\Nset}{\mathbb{N}}

\newcommand{\lj}{\l\jop}

\newcommand{\lm}{\Lambda_{[\cdot]}}

\renewcommand{\l}{\lambda}

\newcommand{\cben}[2]{{\color{black} {#2}}}
\newcommand{\ben}[1]{{\color{black} {#1}}}

\newcommand{\comments}[1]{}

\newcommand{\ie}{{\em i.e.}}
\newcommand{\ih}{\textit{i.h.}}
\newcommand{\fv}[1]{{\tt fv}(#1)}

\newcommand{\jop}{{\tt j}}

\newcommand{\hole}{[\cdot]}

\newcommand{\set}[1]{\{#1\}}

 \newcommand{\myinput}[1]{
 \ifthenelse{\boolean{withimages}}{\input{#1}}{}
 }

\newcommand{\lssym}{{\tt ls}}

\newcommand{\ssym}{{\tt s}}
\newcommand{\tos}{\Rew{\lssym}}
\newcommand{\tohl}{\multimap}

\newcommand{\hlb}{{\tt dB}}

\newcommand{\linunf}[1]{#1\rotatebox[origin=c]{-90}{$\multimap$}}

\newcommand{\rtohls}{\hat{\tohls}}
\newcommand{\tohlb}{\multimap_{\hlb}}
\newcommand{\tohls}{\multimap_{\lssym}}

\newcommand{\hhp}[1]{\hh[#1]}
\newcommand{\hh}{HH}
\newcommand{\hhmes}[1]{|#1|_{\hh}}

\renewcommand{\L}{{\tt L}}

\long\def\ignore#1{\relax}

\newcommand{\ShaneIgnore}[1]{}

\def\itmath#1{\leavevmode\ifmmode{\mbox{\it#1} }\else{\it#1 }\fi}
\def\sfmath#1{\leavevmode\ifmmode{\mbox{\sf#1} }\else{\sf#1 }\fi}
\def\condmath#1{\leavevmode\ifmmode{#1}\else{$#1$}\fi}

  {\begin{enumerate}%
    \setlength{\itemsep}{0pt}%
    \setlength{\parskip}{0pt}%
    \setlength{\topsep}{0pt}
    \setlength{\partopsep}{0pt}}
  {\end{enumerate}}

  {\begin{description}%
    \setlength{\itemsep}{0.1pt}%
    \setlength{\parskip}{0pt}%
    \setlength{\topsep}{0pt}
    \setlength{\partopsep}{0pt}}
  {\end{description}}

\newcommand{\sep}{\hspace*{0.5cm}}

\def\l{\lambda}

\newcommand{\Rew}[1]{\rightarrow_{#1}}

\newcommand{\parunf}[2]{#1\rotatebox[origin=c]{-90}{$\rightarrow$}_{#2}}
\newcommand{\totunf}[1]{#1\rotatebox[origin=c]{-90}{$\rightarrow$}}

\newcommand{\isubs}[1]{\{#1\}}

\newcommand{\nat}{{\mathbb N}}

\newcommand{\B}{{\tt dB}}

\def\l{\lambda}

\newcommand{\M}{{\cal M}}

\newcommand{\tob}{\Rew{\beta}}

\newcommand{\terms}{{\cal T}}
\newcommand{\lterms}{\terms_\l}

\newcommand{\valone}{v}
\newcommand{\valtwo}{w}
\newcommand{\btmone}{a}
\newcommand{\btmtwo}{b}
\newcommand{\btmthree}{c}
\newcommand{\btmfour}{d}
\newcommand{\ctxone}{C}
\newcommand{\ctxtwo}{D}
\newcommand{\ctxthree}{E}
\newcommand{\ctxfour}{F}
\newcommand{\conone}{A}
\newcommand{\contwo}{B}
\newcommand{\pairone}{P}

\newcommand{\ctx}[2]{#1[#2]}
\newcommand{\matr}{\mathtt{M}}

\newcommand{\unf}[1]{\overset{#1}{\sim}}
\newcommand{\valop}{\diamond}

\newcommand{\rvar}{{\tt var}}
\newcommand{\rabsoc}{\l_1}
\newcommand{\rabsweak}{\l_2}
\newcommand{\rabsinc}{\l_3}
\newcommand{\rapp}{@}
\newcommand{\rsubl}{{\tt sub}_l}
\newcommand{\rsubr}{{\tt sub}_r}
\newcommand{\runfl}{{\tt unf}_l}
\newcommand{\runfr}{{\tt unf}_r}
\newcommand{\rerrla}{{\tt err}_{\l @}}
\newcommand{\rerral}{{\tt err}_{@\l}}
\newcommand{\rerrlv}{{\tt err}_{\l x}}
\newcommand{\rerrvl}{{\tt err}_{x\l}}
\newcommand{\rerrva}{{\tt err}_{x @}}
\newcommand{\rerrav}{{\tt err}_{@ x}}

\definecolor{LightGray}{gray}{.80}
\definecolor{DarkGray}{gray}{.60}

\newcommand{\chiaro}[1]{\colorbox{LightGray}{\ensuremath{#1}}}
\newcommand{\scuro}[1]{\colorbox{DarkGray}{\ensuremath{#1}}}

\newcommand{\varone}{x}
\newcommand{\vartwo}{y}
\newcommand{\varthree}{z}
\newcommand{\esub}[3]{#1[#2/#3]}
\newcommand{\tmone}{t}
\newcommand{\tmtwo}{u}
\newcommand{\tmthree}{r}
\newcommand{\tmfour}{s}
\newcommand{\tmfive}{v}
\newcommand{\tmsix}{w}
\newcommand{\ltone}{\L}

\newcommand{\econe}{H}

\newcommand{\lecone}{\hat{H}}
\newcommand{\econep}[1]{H[#1]}
\newcommand{\redone}{\rho}
\newcommand{\redtwo}{\tau}
\newcommand{\redthree}{\gamma}

\newcommand{\toh}{\to_{\tt h}}
\newcommand{\sizeb}[1]{|#1|_{\B}}
\newcommand{\isub}[3]{#1\set{#2/#3}}

\newcommand{\tonhls}{\Rightarrow_{\ssym}}

\newcommand{\esnum}[1]{{\tt es}(#1)}

\newcommand{\coh}{\sim}

\newcommand{\rtop}[1]{\mapsto_{#1}}
\newcommand{\rtob}{\rtop{\beta}}
\newcommand{\rtodb}{\rtop{\B}}
\newcommand{\rtos}{\rtop{\lssym}}
\newcommand{\rtow}{\rtop{{\tt gc}}}

\newenvironment{varitemize}
{
\begin{list}{\labelitemi}
{\setlength{\itemsep}{0pt}
 \setlength{\topsep}{0pt}
 \setlength{\parsep}{0pt}
 \setlength{\partopsep}{0pt}
 \setlength{\leftmargin}{15pt}
 \setlength{\rightmargin}{0pt}
 \setlength{\itemindent}{0pt}
 \setlength{\labelsep}{5pt}
 \setlength{\labelwidth}{10pt}
}}
{
 \end{list} 
}

\newcounter{numberone}

\newenvironment{varenumerate}
{
\begin{list}{\arabic{numberone}.}
{
  \usecounter{numberone}
  \setlength{\itemsep}{0pt}
  \setlength{\topsep}{0pt}
  \setlength{\parsep}{0pt}
  \setlength{\partopsep}{0pt}
  \setlength{\leftmargin}{15pt}
  \setlength{\rightmargin}{0pt}
  \setlength{\itemindent}{0pt}
  \setlength{\labelsep}{5pt}
  \setlength{\labelwidth}{15pt}
}}
{
\end{list} 
}

\newcommand{\hsym}{{\tt h}}
\newcommand{\tolm}{\Rew{\lm}}
\newcommand{\todb}{\Rew{\B}}
\newcommand{\tow}{\Rew{{\tt gc}}}
\newcommand{\tounf}{\Rew{\ssym}}
\newcommand{\downtohl}{\rotatebox[origin=c]{-90}{$\multimap$}}

\newcommand{\toone}{\Rew{1}}
\newcommand{\totwo}{\Rew{2}}
\newcommand{\toonetwo}{\Rew{1,2}}
\newcommand{\diagarrow}{\rotatebox[origin=c]{-135}{$\leadsto$}}

\condinc
{
\author{Beniamino Accattoli \and Ugo Dal Lago}
\title{On the Invariance of the Unitary Cost Model\\ for Head Reduction\\ (Long Version)}
\date{}
}
{
\title{On the Invariance of the Unitary Cost Model for Head Reduction\footnote{This work was partially supported by the ARC INRIA project
``ETERNAL''.}}
\titlerunning{On the Invariance of the Unitary Cost Model for Head Reduction}

\author[1]{Beniamino Accattoli}
\author[2]{Ugo Dal Lago}
\affil[1]{INRIA \& LIX (\'Ecole Polytechnique)\\
  \ignore{91128 Palaiseau Cedex,
France\\}
  \texttt{beniamino.accattoli@inria.fr}}
\affil[2]{Dipartimento di Scienze dell'Informazione, Universit\`a di Bologna\\
  \texttt{dallago@cs.unibo.it}}
\authorrunning{B. Accattoli and U. Dal Lago}

\Copyright[nc-nd]
          {B. Accattoli and U. Dal Lago}

\serieslogo{rta-logo}
\volumeinfo
  {A. Tiwari Editor}
  {1}
  {submitted to 23rd International Conference on Rewriting Techniques and
Applications}
  {1}
  {1}
  {1}
\EventShortName{RTA'12}
\DOI{10.4230/LIPIcs.xxx.yyy.p}

\theoremstyle{plain}\newtheorem{proposition}[theorem]{Proposition}
}

\begin{document}

\maketitle
 
\begin{abstract}
The $\lambda$-calculus is a widely accepted computational model of higher-order
functional programs, yet there is not any direct and universally accepted
cost model for it. As a consequence, the computational difficulty
of reducing $\lambda$-terms to their normal form is typically studied
by reasoning on concrete implementation algorithms. 
In this paper, we show that when head reduction is the underlying 
dynamics, the unitary cost model is indeed invariant.
This improves on known results, which only deal with weak 
(call-by-value or call-by-name) reduction. Invariance
is proved by way of a \cben{calculus of \emph{linear}}{\emph{linear} calculus of} explicit substitutions, which allows
to nicely decompose any head reduction step in the $\lambda$-calculus into
more elementary substitution steps, thus making the combinatorics of
head-reduction easier to reason about. The technique is also a promising
tool to attack what we see as the main open problem, namely understanding
for which \emph{normalizing} strategies derivation complexity is an invariant
cost model, if any.
\end{abstract}

\condinc{}{\rtacat{Regular Research Paper}}

\section{Introduction}
Giving an estimate of the amount of time $T$ needed to execute a program 
is a natural refinement of the termination problem, which only requires
to decide whether $T$ is either finite or infinite. The shift from termination 
to complexity analysis  brings more informative outcomes at the price of an 
increased difficulty. In particular, complexity analysis depends much on 
the chosen computational model. Is it possible to express such estimates in a 
way which is independent from the specific machine the program is run on?
An answer to this question can be given following computational complexity, which
classifies functions based on the amount of time (or space) they consume 
when executed by \emph{any} abstract device endowed with a \emph{reasonable} cost model, depending on the size of input. 
When can a cost model be considered reasonable? The answer
lies in the so-called invariance thesis~\cite{vanEmdeBoas90}: any time cost model is
reasonable if it is polynomially related to the (standard) one of Turing machines.

If programs are expressed as rewrite systems (e.g. as first-order TRSs), an
abstract but effective way to execute programs, rewriting itself, is always
available. As a consequence, a natural time cost model turns out to be \emph{derivational
complexity}, namely the (maximum) number of rewrite steps which can possibly be
performed from the given term. A rewriting step, however, may not 
be an atomic operation, so derivational complexity is not by definition
invariant. For first-order TRSs, however, derivational complexity
has been recently shown to be an invariant cost model, by way of term graph
rewriting~\cite{DalLagoM10,AvanziniM10}.

The case of $\lambda$-calculus is definitely more delicate: if $\beta$-reduction
is weak, i.e., if it cannot take place in the scope of $\lambda$-abstractions, one
can see $\lambda$-calculus as a TRS and get invariance by way of the already cited results~\cite{DalLagoM09}, or
by other means~\cite{SandsGM02}. But if one needs to reduce ``under lambdas''
because the final term needs to be in normal form (e.g., when performing type checking in dependent type
theories), no invariance results are known at the time of writing.

In this paper we give a partial solution to this problem, by showing that
the unitary cost model is indeed invariant for the $\lambda$-calculus
endowed with \emph{head reduction}, in which reduction \emph{can} take place
in the scope of $\lambda$-abstractions, but \emph{can only} be performed
in head position. Our proof technique consists in implementing head reduction in  a calculus of explicit substitutions. 

Explicit substitutions were introduced to close the gap between the theory of $\l$-calculus and implementations 
\cite{ACCL91es}. Their rewriting theory has also been studied in depth, after \mellies\ showed the possibility of 
pathological behaviors \cite{DBLP:conf/tlca/Mellie95}. Starting from graphical syntaxes, a new \emph{at a distance} approach to 
explicit substitutions has recently been proposed \cite{AccattoliK10}. The new formalisms are simpler than those of 
the earlier generation, and another thread of applications --- to which this paper belongs --- also started: 
new results on $\l$-calculus have been proved by means of explicit substitutions \cite{AccattoliK10,AKLPAR}. 

In this paper we use the \emph{linear-substitution calculus} $\lm$, a slight variation over a calculus of explicit substitutions introduced 
by Robin Milner \cite{Milner2006}. The variation is inspired by the structural 
$\l$-calculus \cite{AccattoliK10}. We study in detail the relation between $\l$-calculus head reduction and \emph{linear head reduction} 
\cite{DBLP:journals/tcs/MascariP94}, the head reduction of $\lm$, and prove that the latter can be at most quadratically 
longer than the former. This is proved without any termination assumption, by a detailed rewriting 
analysis. 

To get the Invariance Theorem, however, other ingredients are required:
\begin{varenumerate}
\item  
  \emph{The Subterm Property}. Linear head reduction has a property not enjoyed by head $\beta$-reduction: 
  linear substitutions along a reduction $\tmone\tohl^* \tmtwo$ duplicates subterms of $\tmone$ only. It easily 
  follows that $\tohl$-steps can be simulated by Turing machines in time polynomial in the size of $t$ and the 
  length of $\tohl^*$. This is explained in Section~\ref{s:expsubst}.
\item 
  \emph{Compact representations}. Explicit substitutions, decomposing $\beta$-reduction into more atomic 
  steps, allow to take advantage of sharing and thus provide compact representations of terms, avoiding the 
  exponential blowups of term size happening in plain $\lambda$-calculus. Is it 
  reasonable to use these compact representations of $\l$-terms? We answer affirmatively, by exhibiting 
  a dynamic programming algorithm for checking equality of terms with explicit substitutions modulo unfolding, 
  and proving it to work in polynomial time in the size of the involved compact representations.
  This is the topic of Section~\ref{s:unfolding}.
\item 
  \emph{Head simulation of Turing machines}. We also provide the simulation of 
  Turing machines by $\lambda$-terms. We give a new encoding of Turing machines, since the known ones do not work 
  with \emph{head} $\beta$-reduction, and prove it induces a polynomial overhead. Some details of the encoding are given
  in Section~\ref{s:tur-mach}.
\end{varenumerate}
We emphasize the result for head $\beta$-reduction, but our technical detour also proves invariance for linear 
head reduction. To our knowledge, we are the firsts to use the fine granularity of explicit substitutions for 
complexity analysis. Many calculi with bounded complexity (e.g. \cite{DBLP:journals/aml/Terui07}) use ${\tt let}$-constructs, an avatar of 
explicit substitutions, but they do not take advantage of the refined dynamics, as they always use 
big-steps substitution rules.

To conclude, we strongly believe that the main contribution of this paper lies in the technique rather 
than in the invariance result. Indeed, the main open problem in this area, namely the invariance of 
the unitary cost model for any \emph{normalizing} strategy remains open but, as we argue in 
Section~\ref{s:relation}, seems now within reach.

\condinc{}{An extended version of this paper with more detailed proofs is available from the authors~\cite{ExtendedVersion}.}
 
\section{$\lambda$-Calculus and Cost Models: an Informal Account}\label{s:informal}
Consider the pure, untyped, $\lambda$-calculus. Terms can be variables, abstractions
or applications and computation is captured by $\beta$-reduction. Once a reduction
strategy is fixed, one could be tempted to make \emph{time} and \emph{reduction steps}
to correspond: firing a $\beta$-redex requires one time instant (or, equivalently,
a finite number of time instants) and thus the number of reduction steps to
normal form could be seen as a measure of its time complexity. This would be very
convenient, since reasoning on the complexity of normalization could be done this
way directly on $\lambda$-terms. However, doing so one could in principle
risk to be too optimistic about the complexity of obtaining the normal form of a term
$\tmone$, given $\tmone$ as an input. This section will articulate on this issue
by giving some examples and pointers to the relevant literature.

Consider the sequence of $\lambda$-terms defined as follows, by induction on a natural
number $n$ (where $\tmtwo$ is the lambda term $\vartwo\varone\varone$): $\tmone_0=\tmtwo$
and for every $n\in\Nset$, $\tmone_{n+1}=(\lambda\varone.\tmone_n)\tmtwo$.
$\tmone_n$ has size linear in $n$, and $\tmone_n$ rewrites to its normal form $\tmthree_n$
in exactly $n$ steps, following a leftmost-outermost strategy:
\begin{align*}
\tmone_0&\equiv\tmtwo\equiv\tmthree_0\\
\tmone_1&\rightarrow\vartwo\tmtwo\tmtwo\equiv\vartwo\tmthree_0\vartwo\tmthree_0\equiv\tmthree_1\\
\tmone_2&\rightarrow(\lambda\varone.\tmone_0)(\vartwo\tmtwo\tmtwo)\equiv(\lambda\varone.\tmtwo)(\tmthree_1)
  \rightarrow\vartwo\tmthree_1\tmthree_1\equiv\tmthree_2\\
&\vdots
\end{align*}
For every $n$, however, $\tmthree_{n+1}$ contains two copies of $\tmthree_n$, hence
the size of $\tmthree_n$ is \emph{exponential} in $n$. As a consequence,
if we stick to the leftmost-outermost strategy and if we insist on normal forms to
be represented explicitly, without taking advantage of sharing, the unitary cost model
\emph{is not} invariant: in a linear number of $\beta$-step we reach
an object which cannot even be written down in polynomial time.

One may wonder whether this problem is due to the specific, inefficient, adopted strategy.
However, it is quite easy to rebuild a sequence of terms exhibiting the same behavior
along an innermost strategy: if $\tmfour=\lambda\vartwo.\vartwo\varone\varone$,
then define $\tmfive_0$ to be just $\lambda\varone.\tmfour$ and
for every $n\in\Nset$, $\tmfive_{n+1}$ to be $\lambda\varone.\tmfive_n\tmfour$.
Actually, there \emph{are} invariant cost-models for the $\lambda$-calculus
even if one wants to obtain the normal form in an explicit, linear format, like the
difference cost model~\cite{DalLagoM08}. But they pay a price for that: 
they do not attribute a constant weight to each reduction
step. Then, another natural question arises: is it that the gap between
the unitary cost model and the real complexity of reducing terms is only due to a 
\emph{representation} problem? In other words, could we take advantage of a shared
representation of terms, even if only to encode $\lambda$-terms
(and normal forms in particular) in a compact way?

The literature offers some positive answers to the question above.
In particular, the unitary cost model can be proved to be invariant for
both call-by-name and call-by-value $\lambda$-calculi, as defined by Plotkin~\cite{DBLP:journals/tcs/Plotkin75}.
In one way or another, the mentioned results are based on sharing subterms,
either by translating the $\lambda$-calculus to a TRS \cite{DalLagoM09} or by going through
abstract machines \cite{SandsGM02}. Plotkin's calculi, however, are endowed with \emph{weak}
notions of reduction, which prevent computation to happen in the scope of
a $\lambda$-abstraction. And the proposed approaches crucially rely on that.

The question now becomes the following: is it possible to prove the invariance
of the unitary cost model for some \emph{strong} notion of reduction\ignore{, or maybe for
a normalizing strategy}? This paper gives a first, positive answer to this question
by proving the number of $\beta$-reduction steps to be an invariant cost model
for \emph{head} reduction, in which one \emph{is} allowed to reduce in the scope
of $\lambda$-abstraction, but evaluation stops on \emph{head} normal forms. 

We are convinced that the exponential blowup in the examples above
is, indeed, only due to the $\lambda$-calculus being a very inefficient
\emph{representation} formalism. Following this thesis we use terms with 
explicit substitutions as compact representations: our approach, in contrast 
to other ones, consists in using sharing (here under the form of explicit 
substitutions) only to obtain compactness, and not to design some 
sort of optimal strategy reducing shared redexes. Actually, we follow the 
opposite direction: the leftmost-outermost strategy --- being standard --- can be considered 
as the maximally \emph{non-sharing} strategy. How much are we losing limiting ourselves 
to head reduction? Not so much: in Section \ref{s:tur-mach} we show an encoding of Turing 
machines for which the normal form is reached by head-reduction only. Moreover, from a 
denotational semantics point of view head-normal forms --- and not full normal forms --- 
are the right notion of result for $\beta$-reduction.

The next two sections introduce explicit substitutions and prove that the length of their head 
strategy is polynomially related to the length of head $\beta$-reduction. In other words, 
the switch to compact representations does not affect the cost model in an essential way.

\section{Linear Explicit Substitutions}\label{s:expsubst}
\label{ss:pure-linear-head}
First of all, we introduce the $\l$-calculus. Its terms are given by the grammar:
$$
\tmone,\tmtwo,\tmthree \in\lterms:: \varone \mid \lterms\ \lterms \mid \l \varone. \lterms 
$$
and its reduction rule $\tob$ is defined as the context closure of
$(\l \varone.\tmone)\ \tmtwo \rtob \isub{\tmone}{\varone}{\tmtwo}$.
We will mainly work with head reduction, instead of full $\beta$-reduction. 
We define head reduction as follows. Let an \emph{head context} $\lecone$ be defined by:
$$
\lecone::=\hole \mid \lecone\ \lterms \mid \l \varone. \lecone
$$
Then define \deff{head reduction} $\toh$ as the closure by head contexts of $\rtob$. 
Our definition of head reduction is slightly more liberal than the usual one. Indeed, 
it is non-deterministic, for instance: 
$$ 
(\l \varone. I)\ \tmone\  _\hsym\leftarrow (\l \varone. (I\ I))\ \tmone \toh I\ I
$$
\condinc{
Usually only one of the two redexes would be considered an head redex. However, this non-determinism is harmless, since one easily proves that
\begin{lemma}
$\toh$ has the diamond property, namely if $\tmone\toh \tmtwo_i$ with $i=1,2$ then there exists $\tmthree$ such that $\tmtwo_i\toh \tmthree$.
\end{lemma}}
{Usually only one of the two redexes would be considered an head redex. However, this non-determinism is harmless, since one easily proves that
$\toh$ has the diamond property.}
Reducing $\toh$ in an outermost way we recover the usual notion of head reduction, so our approach gains in generality without loosing any property 
of head reduction. Our notion is motivated by the corresponding notion of head reduction for explicit substitutions, which is easier to manage in this more general approach.

The calculus of explicit substitutions we are going to use is a minor variation over a simple calculus introduced by Milner \cite{Milner2006}. The grammar is standard:
$$
\tmone,\tmtwo,\tmthree \in\terms:: \varone \mid \terms\ \terms \mid \l \varone. \terms \mid \esub{\terms}{\varone}{\terms}
$$ 
The term $\esub{\tmone}{\varone}{\tmtwo}$ is an \deff{explicit substitution}. Both costructors $\l \varone. \tmone$ and $\esub{\tmone}{\varone}{\tmtwo}$ 
bind $\varone$ in $\tmone$. \ben{We note $\ltone$ a possibly empty list of explicit substitutions $\esub{}{\varone_1}{\tmtwo_1}\ldots \esub{}{\varone_k}{\tmtwo_k}$.} 
Contexts are defined by:
$$
\ctxone,\ctxtwo,\ctxthree,\ctxfour:: \hole \mid \ctxone\ \terms \mid \terms\ \ctxone \mid \l \varone. \ctxone \mid \esub{\ctxone}{\varone}{\terms} \mid \esub{\terms}{\varone}{\ctxone}
$$
We note $\ctx{\ctxone}{\tmone}$ the usual operation of substituting $\hole$ in $\ctxone$ possibly capturing 
free variables of $\tmone$. We will often use expressions like $\esub{\ctx{\ctxone}{\varone}}{\varone}{\tmtwo}$ 
where it is implicitly assumed that $\ctxone$ does not capture $\varone$. 
The \deff{linear-substitution calculus} $\lm$ is given by the rewriting rules $\todb$, $\tos$ and $\tow$, defined as 
the context closure of the rules $\rtodb$, $\rtos$ and $\rtow$ in Figure \ref{fig:lmrules}. We also use the notation 
$\tolm= \todb\cup\tos\cup\tow$ and $\tounf=\tos\cup\tow$. \ben{Rule $\todb$ acts at a \emph{distance}: the function 
$\l \varone.\tmone$ and the argument $\tmtwo$ can interact even if there is $\ltone$ between them. This is motivated 
by the close relation between $\lm$ and graphical formalisms as proof-nets and $\lj$-dags \cite{AccattoliTh, AccattoliG09}, 
and is also the difference with Milner's presentation of $\lm$ \cite{Milner2006}.}

The linear-substitution calculus enjoys all properties required to explicit substitutions calculi, obtained by easy 
adaptations of the proofs for $\lj$ in $\cite{AccattoliK10}$. Moreover, it is confluent and preserves $\beta$-strong 
normalization. In particular, $\tounf$ is a strongly normalizing and confluent relation. 

Given a term $\tmone$ with explicit substitutions, its normal form with respect to $\tounf$ is a $\l$-term, noted $\totunf{\tmone}$, 
called the \deff{unfolding} of $\tmone$ and verifying the following equalities:
\[ \begin{array}{l@{\sep\sep}l@{\sep\sep}l@{\sep\sep}l}
   \totunf{(\tmone\ \tmtwo)}  = \totunf{\tmone}\ \totunf{\tmtwo}&
   \totunf{(\l \varone. \tmone)}  =  \l \varone. \totunf{\tmone} &
   \totunf{(\esub{\tmone}{\varone}{\tmtwo})}  =  \isub{\totunf{\tmone}}{\varone}{\totunf{\tmtwo}}
   \end{array} \]
Another useful property is the so-called \emph{full-composition}, which states that any explicit substitution can be 
reduced to its implicit form independently from the other substitutions in the term, formally 
$\esub{\tmone}{\varone}{\tmtwo}\tounf^* \isub{\tmone}{\varone}{\tmtwo}$. Last, $\lm$ simulates $\l$-calculus 
($\tmone\tob\tmtwo$ implies $\tmone\tolm^*\tmtwo$) and reductions in $\lm$ can be projected on $\l$-calculus 
via unfolding ($\tmone\tolm\tmtwo$ implies $\totunf{\tmone}\tob^*\totunf{\tmtwo}$).

The calculus $\lm$ has a strong relation with proof-nets and linear logic: 
it can be mapped to Danos' and Regnier's pure proof-nets~\cite{Reg:Thesis:92} 
or to $\lj$-dags \cite{AccattoliG09}. The rule $\todb$ corresponds to proof-nets 
multiplicative cut-elimination, $\tos$ to the cut-elimination rule between $!$ 
(every substitution is in a $!$-box) and contraction, $\tow$ to the cut-elimination 
rule between $!$ and weakening. The case of a cut between $!$ and dereliction is 
handled by $\tos$, as if cut derelictions were always contracted with a weakening. 

\begin{figure}
$$ \begin{array}{llllll}
    (\l \varone.\tmone)\ltone \; \tmtwo & \rtodb & \esub{\tmone}{\varone}{\tmtwo}\ltone & \\
    \esub{\ctx{\ctxone}{\varone}}{\varone}{\tmtwo} & \rtos&  \esub{\ctx{\ctxone}{\tmtwo}}{\varone}{\tmtwo}\\
    \esub{\tmone}{\varone}{\tmtwo} & \rtow&  \tmone & \mbox{ if } \varone\notin\fv{\tmone} \\
  \end{array}$$
  \caption{$\lm$ rewriting rules\label{fig:lmrules}}
\end{figure}

\subsection{Linear Head Reduction, the Subterm Property and Shallow Terms}
In this paper, we mainly deal with a specific notion of reduction for 
$\lm$, called \deff{linear head reduction} \cite{DBLP:journals/tcs/MascariP94,Danos04headlinear},
and related to the representation of $\l$-calculus into linear logic proof-nets. 
In order to define it we need the notion of head context for explicit substitutions.
\begin{definition}[head context]
\deff{Head contexts} are defined by the following grammar:
$$
\econe::=\hole \mid \econe\ \terms \mid \l \varone. \econe\mid\esub{\econe}{\varone}{\terms}.
$$
\end{definition}
The fundamental property of an head context $\econe$ is that the hole cannot be duplicated 
nor erased. In terms of linear logic proof-nets, the head hole is is not contained in any 
box (since boxes are associated with the arguments of  applications and with explicit substitutions).
We now need to consider a modified version of $\rtos$:
$$\begin{array}{cccc}
  \esub{\econep{\varone}}{\varone}{\tmtwo}&\rtohls &\esub{\econep{\tmtwo}}{\varone}{\tmtwo}
\end{array}$$
Now, let $\tohlb$ (resp. $\tohls$) be the closure by head contexts of $\rtodb$ (resp. $\rtohls$). 
Last, define \deff{head linear reduction} $\tohl$ as $\tohlb\cup\tohls$.
Please notice that $\tohl$ can reduce under $\l$, for instance $\l y.(H[x][x/u])\tohls\l y.(H[u][x/u])$. 
Our definition of $\tohl$ gives a non-deterministic strategy, but its non-determinism is again harmless\condinc{:
\begin{lemma}
$\tohl$ has the diamond property.
\end{lemma}

\begin{proof}
The only possible critical pair is given by 
\[\begin{array}{ccccccccc}
H_1[\scuro{(\l y.\chiaro{H_2[x]})\L\ v}]\chiaro{[x/u]} && \tohl&& H_1[\scuro{(\l y.H_2[u])\L\ v}][x/u] \\
\downtohl &&& &\downtohl\\
H_1[\chiaro{H_2[x]}[y/v]\L]\chiaro{[x/u]} && \tohl && H_1[H_2[u][y/v]\L][x/u]
\end{array}\]
There is also the possibility that a $\tohl$-redex $R$ is contained into another $\tohl$-redex $R'$, but then $R'$ cannot duplicate $R$:
\[\begin{array}{ccccccccc}
H_1[\scuro{(\l y.H_2[\chiaro{(\l x.t)\L_2\ u}])\L_1\ v}] && \tohl&& H_1[H_2[\chiaro{(\l x.t)\L_2\ u}][y/v]\L_1]\\
\downtohl &&&&\downtohl\\
H_1[\scuro{(\l y.H_2[t[x/u]\L_2])\L_1\ v}] && \tohl && H_1[H_2[t[x/u]\L_2][y/v]\L_1]
\end{array}\]
There are no other cases, in particular no term can have two disjoint $\tohl$-redexes.
\end{proof}
}{: a simple case analysis shows that $\tohl$ has the diamond property.}

A term $\tmtwo$ is a \emph{box-subterm} of a term $\tmone$ (resp. of a context $\ctxone$) if $\tmone$ (resp. $\ctxone$) 
has a subterm of the form $\tmthree\ \tmtwo$ or of the form $\tmthree [x/\tmtwo]$ for some $r$.
\begin{remark}
By definition of head-contexts, $\hole$ is not a box-subterm of $\econep{\cdot}$, 
and there is no box-subterm of $\econep{\cdot}$ which has $\hole$ as a subterm.
\end{remark}
\begin{proposition}[Subterm Property]
\label{l:box-subterm}
If $\tmone\tohl^*\tmtwo$ and $\tmthree$ is a box-suterm of $\tmtwo$, then $\tmthree$ is a box-subterm of $\tmone$.
\end{proposition}
The aforementioned proposition is a key point in our study of cost models. Linear head substitution 
steps duplicate sub-terms, but the Subterm Property guarantees that only sub-terms of the initial 
term $\tmone$ are duplicated, and thus each step can be implemented in time polynomial in 
the size of $\tmone$, which is the size of the input, the fundamental parameter for complexity 
analysis. This is in sharp contrast with what happens in the $\l$-calculus, where the cost of a 
$\beta$-reduction step is not even polynomially related to the size of the initial term. 
\condinc{
\begin{proof}
By induction on the length $k$ of the reduction $\tmone\tohl^*\tmtwo$. Suppose $k>0$. Then $\tmone\tohl^*\tmfive\tohl \tmtwo$ and by \ih\ any box-subterm of $\tmfive$ is a box subterm of $\tmone$, so it is enough to show that any box-subterm of $\tmtwo$ is a box-subterm of $\tmfive$. By induction on $\tmfive\tohl \tmtwo$:
\begin{varitemize}
\item Base cases:
\begin{varitemize}
\item $\tmfive=(\l\varone.\tmthree)\ltone\ \tmfour\tohl \esub{\tmthree}{\varone}{\tmfour}\ltone=\tmtwo$: it is evident that the two terms have the same box-subterms. 
\item $\tmfive=\esub{\econep{\varone}}{\varone}{\tmfour}\tohl \esub{\econep{\tmfour}}{\varone}{\tmfour}=\tmtwo$: by a previous remark the plug of $\tmfour$ in $\econep{\cdot}$ does not create any new box-subterm, nor modify a box-subterm of $\econep{\cdot}$. And obviously any box-subterm of $\tmfour$ is a box-subterm of $\tmfive$.
\end{varitemize}
\item Inductive cases: just use the \ih\ and the previous remark.
\end{varitemize}
\end{proof}
}
{
\begin{proof}
Let $\tmone\tohl^k\tmtwo$ By induction on $k$. Suppose $k>0$. Then $\tmone\tohl^*\tmfive\tohl \tmtwo$ and by \ih\ any box-subterm of $\tmfive$ is a box subterm of $\tmone$, so it is enough to show that any box-subterm of $\tmtwo$ is a box-subterm of $\tmfive$. By induction on $\tmfive\tohl \tmtwo$. If $\tmfive=(\l\varone.\tmthree)\ltone\ \tmfour\tohl \esub{\tmthree}{\varone}{\tmfour}\ltone$ is evident. If $\tmfive=\esub{\econep{\varone}}{\varone}{\tmfour}\tohl \esub{\econep{\tmfour}}{\varone}{\tmfour}$ by the previous remark the plug of $\tmfour$ in $\econep{\cdot}$ does not create any new box-subterm, nor modify a box-subterm of $\econep{\cdot}$. And obviously any box-subterm of $\tmfour$ is a box-subterm of $\tmfive$. The inductive cases follow from the \ih\ and the previous remark.
\end{proof}
}
The subterm property does not only concern the cost of implementing reduction steps, but also the size of the end term:
\begin{corollary}
There is a polynomial $p:\Nset\times\Nset\rightarrow\Nset$ such that
if $\tmone\tohl^k\tmtwo$ then $|\tmtwo|\leq p(k,|\tmone|)$.
\end{corollary}
Consider a reduction $\tmone\tohl^*\tmtwo$ where $\tmone$ is a $\l$-term. Another consequence of the Subterm Property is 
that for every explicit substitution occurring in $\tmtwo$, the substituted term is a $\l$-term. This is another
strong property to be used in the analysis of the next section.
\begin{definition}[Shallow Terms]
\label{d:shallow}
A $\lm$-term $\tmone$ is \deff{shallow} if whenever $\tmone=\ctx{\ctxone}{\esub{\tmtwo}{\varone}{\tmthree}}$ then $\tmthree$ is a $\l$-term.
\end{definition}
\begin{corollary}
Let $\tmone$ be a $\l$-term and $\tmone\tohl^*\tmtwo$. Then $\tmtwo$ is shallow.
\end{corollary}
\condinc{
\begin{proof}
By lemma \ref{l:box-subterm} the content of any substitution of $\tmtwo$ is a subterm of $\tmone$, \ie\ a $\l$-term.
\end{proof}
}{}

\section{On the Relation Between $\Lambda$ and $\Lambda_{[\cdot]}$}\label{s:relation}
In this section, linear explicit substitutions will be showed to be an efficient way to implement
head reduction. We will proceed by proving three auxiliary results separately: 
\begin{varenumerate}
\item 
  We show that any $\tohl$-reduction $\rho$ projects via unfolding to a $\toh$-reduction 
  $\totunf{\rho}$ having as length exactly the number of $\tohlb$ steps in $\rho$; this
  is the topic of Section~\ref{ss:projlonh};
\item 
  We show the converse relation, \ie\ that any $\toh$-reduction $\rho$ can be simulated 
  by a $\tohl$-reduction having as many $\tohlb$-steps as the the steps in $\rho$, followed by unfolding;
  this is in Section~\ref{ss:projhonl};
\item
  We show that in any $\tohl$-reduction $\rho$ the number of $\tohls$-steps is $\mathcal{O}(\sizeb{\rho}^2)$ 
  where $\sizeb{\rho}$ is the number of $\tohlb$ steps in $\rho$. By the previous two points, there is 
  a quadratic --- and thus polynomial --- relation between $\toh$-reductions and $\tohl$-reduction from a 
  given term; all this is explained in Section~\ref{ss:quadratic}.
\end{varenumerate}
\subsection{Projection of $\tohl$ on $\toh$}\label{ss:projlonh}
The first thing one needs to prove about head-linear reduction is whether it is a \emph{sound} way to
implement head reduction. This is proved by relatively standard techniques, and requires
the following auxiliary lemma, whose proof is by induction on $\tmone\toh\tmtwo$:
\begin{lemma}
\label{l:aux-lh-to-h}
Let $\tmone\in\lterms$. If $\tmone\toh\tmtwo$ then $\isub{\tmone}{\varone}{\tmthree}\toh\isub{\tmtwo}{\varone}{\tmthree}$.
\end{lemma}

\begin{lemma}[Projection of $\tohl$ on $\toh$]
\label{l:hl-to-h}
Let $\tmone\in\terms$. If $\redone: \tmone\tohl^k \tmtwo$ then $\totunf{\tmone}\toh^n \totunf{\tmtwo}$ and $n=\sizeb{\redone}\leq k$.
\end{lemma}
\condinc{
\begin{proof}
By induction on $k$. If $k=0$ it is trivial, so let $k>0$ and $\tmone \tohl^{n-1}\tmthree\tohl \tmtwo$. Let $\redtwo$ be the reduction $\tmone \tohl^{n-1}\tmthree$. By \ih\ we get $\totunf{\tmone}\toh^m \totunf{\tmthree}$ and $m=\sizeb{\redtwo}\leq k-1$. Now consider $\tmthree\tohl \tmtwo$. There are two cases:
\begin{varitemize}
\item If $\tmthree\tohls \tmtwo$ then $\totunf{\tmthree}=\totunf{\tmtwo}$, by definition of $\totunf{(\cdot)}$, as the normal form of $\tounf$ (which contains $\tohls$).

\item If $\tmthree\tohlb \tmtwo$ then $\tmthree=\econep{(\l\varone.\tmfive)\ltone\ \tmfour}\tohl \econep{\esub{\tmfive}{\varone}{\tmfour}\ltone}=\tmtwo$. We prove that $\totunf{\tmthree}\toh\totunf{\tmtwo}$, from which it follows that $\totunf{\tmone}\toh^{m+1}\totunf{\tmtwo}$, where $m+1=\sizeb{\redtwo}+1=\sizeb{\redone}\leq k$.\\
By induction on $\econe$. We use the following notation: if $\tmone=\tmtwo\ltone$ and $\ltone= \esub{}{\vartwo_1}{\tmsix_1}\ldots \esub{}{\vartwo_m}{\tmsix_m}$ then we write $\sigma_\ltone$ for $\isub{}{\vartwo_1}{\totunf{\tmsix_1}}\ldots \isub{}{\vartwo_m}{\totunf{\tmsix_m}}$, thus we can write $\totunf{\tmone}=\totunf{\tmtwo}\sigma_\ltone$. Cases:
\begin{varitemize}
 \item $\econe=\hole$. Then:

\[\begin{array}{lllllll}
\totunf{\tmthree}&=&\totunf{((\l\varone.\tmfive)\ltone)}\ \totunf{\tmfour}&= &&\mbox{(def. of $\totunf{(\cdot)}$)}\\
&&(\l\varone.\totunf{\tmfive})\sigma_\ltone\ \totunf{\tmfour}&= \\
&&(\l\varone.\totunf{\tmfive}\sigma_\ltone)\ \totunf{\tmfour}&\toh \\
&&\isub{\totunf{\tmfive}\sigma_\ltone}{\varone}{\totunf{\tmfour}}&= && \mbox{$\varone\notin\fv{\tmsix_i}$ for $i\in\set{1,\ldots,m}$}\\
&&\isub{\totunf{\tmfive}}{\varone}{\totunf{\tmfour}}\sigma_\ltone&= &&\mbox{(def. of $\totunf{(\cdot)}$)}\\
&&\totunf{(\esub{\tmfive}{\varone}{\tmfour})}\sigma_\ltone&= &&\mbox{(def. of $\totunf{(\cdot)}$)}\\
&&\totunf{(\esub{\tmfive}{\varone}{\tmfour}\ltone)}&= \totunf{\tmtwo}\\
\end{array}\]

\item $\econe= \econe'\ \tmsix$. Then 
\[\begin{array}{llllll}
\totunf{\tmthree}&=&\totunf{(\econe'[(\l\varone.\tmfive)\ltone\ \tmfour])}\ \totunf{\tmsix}&\toh& &(\ih)\\
&&\totunf{(\econe'[\esub{\tmfive}{\varone}{\tmfour}\ltone])}\ \totunf{\tmsix}&=& \totunf{\tmtwo}
\end{array}\]

\item $\econe= \l \vartwo. \econe'$. Then 
\[\begin{array}{llllll}
\totunf{\tmthree}&=&\l \vartwo. \totunf{(\econe'[(\l\varone.\tmfive)\ltone\ \tmfour])}&\toh& &(\ih)\\
&&\l \vartwo.\totunf{(\econe'[\esub{\tmfive}{\varone}{\tmfour}\ltone])}&=& \totunf{\tmtwo}
\end{array}\]

\item $\econe= \esub{\econe'}{\vartwo}{\tmsix}$. Then 
\[\begin{array}{llllll}
\totunf{\tmthree}&=& \isub{\totunf{(\econe'[(\l\varone.\tmfive)\ltone\ \tmfour])}}{\vartwo}{\totunf{\tmsix}}&\toh& &(\ih\ \&\ l. \ref{l:aux-lh-to-h})\\
&&\isub{\totunf{(\econe'[\esub{\tmfive}{\varone}{\tmfour}\ltone])}}{\vartwo}{\totunf{\tmsix}}&=& \totunf{\tmtwo}
\end{array}\]

\end{varitemize}
\end{varitemize}
\end{proof}
}
{

\begin{proof}
By induction on $k$. If $k=0$ it is trivial, so let $k>0$ and $\tmone \tohl^{n-1}\tmthree\tohl \tmtwo$. Let $\redtwo$ be the reduction $\tmone \tohl^{n-1}\tmthree$. By \ih\ we get $\totunf{\tmone}\toh^m \totunf{\tmthree}$ and $m=\sizeb{\redtwo}\leq k-1$. Now consider $\tmthree\tohl \tmtwo$. There are two cases. If $\tmthree\tohls \tmtwo$ then $\totunf{\tmthree}=\totunf{\tmtwo}$, by definition of $\totunf{(\cdot)}$, as the normal form of $\tounf$ (which contains $\tohls$). If $\tmthree\tohlb \tmtwo$ then $\tmthree=\econep{(\l\varone.\tmfive)\ltone\ \tmfour}\tohl \econep{\esub{\tmfive}{\varone}{\tmfour}\ltone}=\tmtwo$. We prove that $\totunf{\tmthree}\toh\totunf{\tmtwo}$, from which it follows that $\totunf{\tmone}\toh^{m+1}\totunf{\tmtwo}$, where $m+1=\sizeb{\redtwo}+1=\sizeb{\redone}\leq k$.\\
By induction on $\econe$. We use the following notation: if $\tmone=\tmtwo\ltone$ and $\ltone= \esub{}{\vartwo_1}{\tmsix_1}\ldots \esub{}{\vartwo_m}{\tmsix_m}$ then we write $\sigma_\ltone$ for $\isub{}{\vartwo_1}{\totunf{\tmsix_1}}\ldots \isub{}{\vartwo_m}{\totunf{\tmsix_m}}$, thus we can write $\totunf{\tmone}=\totunf{\tmtwo}\sigma_\ltone$. Cases:
\begin{varitemize}
\item $\econe=\hole$. Then:
\[\begin{array}{llllllllllllllll}
\totunf{\tmthree}&=&\totunf{((\l\varone.\tmfive)\ltone)}\ \totunf{\tmfour}&=
&(\l\varone.\totunf{\tmfive})\sigma_\ltone\ \totunf{\tmfour}&= 
&(\l\varone.\totunf{\tmfive}\sigma_\ltone)\ \totunf{\tmfour}&\toh \\
&&\isub{\totunf{\tmfive}\sigma_\ltone}{\varone}{\totunf{\tmfour}}&=_{*} 
&\isub{\totunf{\tmfive}}{\varone}{\totunf{\tmfour}}\sigma_\ltone&= 
&\totunf{(\esub{\tmfive}{\varone}{\tmfour})}\sigma_\ltone&= &\totunf{(\esub{\tmfive}{\varone}{\tmfour}\ltone)}&= \totunf{\tmtwo}
\end{array}\]
Where step (*) holds because $\varone\notin\fv{\tmsix_i}$ for $i\in\set{1,\ldots,m}$.

\item $\econe= \esub{\econe'}{\vartwo}{\tmsix}$. Then by \ih\ and Lemma \ref{l:aux-lh-to-h} we get
$\totunf{\tmthree}= \isub{\totunf{(\econe'[(\l\varone.\tmfive)\ltone\ \tmfour])}}{\vartwo}{\totunf{\tmsix}}\toh
\isub{\totunf{(\econe'[\esub{\tmfive}{\varone}{\tmfour}\ltone])}}{\vartwo}{\totunf{\tmsix}}= \totunf{\tmtwo}$

\item The cases $\econe= \econe'\ \tmsix$ and $\econe= \l \vartwo. \econe'$ follow from the \ih

\end{varitemize}
\end{proof}
}

\subsection{Projection of $\toh$ on $\tohl$}\label{ss:projhonl}
Here we map head $\beta$-steps to head linear steps followed by unfolding. In other words, we prove that head-linear reduction
is not only a sound, but also a \emph{complete} way to implement head reduction. This section is going to be 
technically more involved than the last one. First of all, we show that a single head step can be simulated by a step of 
$\tohlb$ followed by unfolding, which is straightforward:
\begin{lemma}[Head Simulation]
\label{l:head-sim}
Let $\tmone$ be a $\l$-term. Then $t\toh u$ then $t\tohlb\tounf^*u$
\end{lemma}
\begin{proof}
By induction on $t\toh u$. The base case: if $(\l \varone.\tmtwo)\ \tmthree\toh \isub{\tmtwo}{\varone}{\tmthree}$ then 
$(\l \varone.\tmtwo)\ \tmthree\tohlb \esub{\tmtwo}{\varone}{\tmthree}$ and 
$\esub{\tmtwo}{\varone}{\tmthree}\tounf^*\isub{\tmtwo}{\varone}{\tmthree}$ by full composition. The inductive cases follows by the \ih.
\end{proof}
We are now going to show that a sequence $\tos^*$ can be factored into some head-linear substitutions, followed
by non-head-linear ones. This allows to refine Lemma~\ref{l:head-sim}.
Define $\tonhls$ as the relation $\tounf\setminus\tohls$, \ie\ $\tonhls$ reduces non-head-linear substitution redexes.
Moreover, define the \deff{linear unfolding} $\linunf{\tmone}$ of $\tmone$ as the normal form of $\tmone$ with respect to 
$\tohls$ (which exists since $\tohls\subseteq\tos$ and $\tos$ terminates, and it is unique because $\tohls$ is deterministic).
\condinc{ We also need the following abstract lemma about postponement of reductions:
\begin{lemma}
\label{l:abstr-postp-cond}
Define $\toonetwo$ as $\toone\cup\totwo$. Then:
\begin{varenumerate}
 \item If $\totwo^* \toone \subseteq \toone^* \totwo^* $ then $\toonetwo^*\subseteq\toone^*\totwo^* $.
\item\label{p:abstr-postp-cond-two} (Geser-Di Cosmo-Piperno) If $\totwo \toone\subseteq \toone^+ \totwo^*$ and $\toone$ is strongly normalising then $\totwo^*\toone^+\subseteq\toone^+\totwo^*$, and so $\toonetwo^*\subseteq\toone^*\totwo^* $.
\end{varenumerate}
\end{lemma}

\begin{proof}
\begin{varenumerate}
\item By induction on the number $k$ of $\toone$ steps in $\tau: t\toonetwo^* u$. The case $k=0$ is trivial. Let $k>0$. Then if $\tau$ is not of the form of the statement it has the following form:
$$t\toone^*\totwo^+\toone\toonetwo^* u$$
Using the hypothesis we get:
$$t\toone^*\toone^*\totwo^*\toonetwo^* u$$
The \ih\ on $\totwo^*\toonetwo^*$ gives:
$$t\toone^*\toone^*\toone^*\totwo^* u$$
and we conclude.
\item We prove the first consequence, the second one follows from the first and Point 1. Let $\tau:t\totwo^k \toone^h u$. By induction on the pair $(\eta(t),k)$, where $\eta(t)$ is the length of the maximal $\toone$ reduction from $t$, ordered lexicographically. The case $k=0$ is trivial, then let $k>0$. Now, if $\tau=\totwo^{k-1}\totwo \toone\toone^{h-1}$ then by applying the first hypothesis to the central subesequence we get $\tau\subseteq\totwo^{k-1}\toone^+\totwo^*\toone^{h-1}$. The measure of the prefix $\totwo^{k-1}\toone^+$ is $(\eta(t),k-1)$ hence by \ih\ we get $\tau\subseteq\toone^+\totwo^*\totwo^*\toone^{h-1}$. For the suffix $\totwo^*\totwo^*\toone^{h-1}$ it is the first component of the measure which decreases, since its starting term is obtained through a $\toone^+$-reduction from $t$, so one can apply the \ih\ and get $\tau\subseteq\toone^+\toone^+\totwo^*$.
\end{varenumerate}
\end{proof}

}{}
Now, we can prove that any $\toh$ step is simulated in $\tohlb\tohls^*\tonhls^*$ 
(actually, in such a sequence there can be at most one $\tohls$ step):
\begin{lemma}[Unfolding Factorization]
\label{l:j-h-and-nh}
The following swaps hold:
\begin{varenumerate}
\item $\tonhls\tohls\subseteq\tohls^+\tonhls^+$, precisely: $\tonhls\tohls\subseteq\tohls\tonhls\cup\tohls\tohls\tonhls\cup \tohls\tonhls\tonhls$.
\item $\tos^*\subseteq\tohls^*\tonhls^*$, and in particular $t\tohls^*\linunf{\tmone}\tonhls^*\totunf{\tmone}$.
\end{varenumerate}
\end{lemma}

\condinc{
\begin{proof}
1) Formally, the proof is by induction on $\tmone\tonhls\tmtwo$, see \cite{ARTATR} for the technical details (and more generally for a detailed study of redex permutations in $\lm$). Informally, $\tonhls$ cannot create, duplicate or alter the head nature of an $\tohls$-step, so that the second step in $\tmone\tonhls\tohls\tmthree$ can be traced back to a unique $\tohls$ redex. Now, there are two cases: 
\begin{varenumerate}
\item The two redexes simply permute.
\item The preponement of $\tohls$ duplicate the redex reduced by $\tonhls$. Two subcases:

\begin{varenumerate}
\item One of the two residuals is a linear head redex. For instance, consider the diagram:
$$\begin{array}{cccccccc}
\esub{\esub{\varone}{\varone}{\vartwo}}{\vartwo}{\varthree}&&\tonhls&&\esub{\esub{\varone}{\varone}{\varthree}}{\vartwo}{\varthree}&\\
\linunf{}_{\lssym}&&\diagarrow&&\linunf{}_{\lssym}\\
\esub{\esub{\vartwo}{\varone}{\vartwo}}{\vartwo}{\varthree}&\tohls&\esub{\esub{\varthree}{\varone}{\vartwo}}{\vartwo}{\varthree}& \tonhls&\esub{\esub{\varthree}{\varone}{\varthree}}{\vartwo}{\varthree}
\end{array}$$

\item Both residuals are non-linear-head redexes:
$$\begin{array}{cccccccc}
\esub{\esub{\varone}{\varone}{\vartwo\ \vartwo}}{\vartwo}{\varthree}&&\tonhls&&\esub{\esub{\varone}{\varone}{\vartwo\ \varthree}}{\vartwo}{\varthree}&\\
\linunf{}_{\lssym}&&\diagarrow&&\linunf{}_{\lssym}\\
\esub{\esub{(\vartwo\ \vartwo)}{\varone}{\vartwo\ \vartwo}}{\vartwo}{\varthree}&\tonhls&\esub{\esub{(\vartwo\ \varthree)}{\varone}{\vartwo\ \vartwo}}{\vartwo}{\varthree}& \tonhls&\esub{\esub{(\vartwo\ \varthree)}{\varone}{\vartwo\ \varthree}}{\vartwo}{\varthree}
\end{array}$$
\end{varenumerate}
\end{varenumerate}

2) Apply Lemma \ref{l:abstr-postp-cond}.2 taking $\toone:=\tohls$, $\totwo:=\tonhls$ and since $\tohls$ is strongly normalizing (from $\tohls\subseteq\tos$ and termination of $\tos$) we get $(\tohls\cup\tonhls)^*\subseteq \tohls^*\tonhls^*$, and conclude since $\tos=\tohls\cup\tonhls$.
\end{proof}
}
{
\begin{proof}
1) Formally, the proof is by induction on $\tmone\tonhls\tmtwo$. Informally, $\tonhls$ cannot create, duplicate or alter the head nature of an $\tohls$-step, so that the second step in $\tmone\tonhls\tohls\tmthree$ can be traced back to a unique $\tohls$ redex. Now, there are two cases: either the two redexes simply permute or the preponement of $\tohls$ duplicate the redex reduced by $\tonhls$. The second case splits in two subcases, of which we just give two examples:
{\small
$$\begin{array}{cccccccc}
\esub{\esub{\varone}{\varone}{\vartwo}}{\vartwo}{\varthree}&&\tonhls&&\esub{\esub{\varone}{\varone}{\varthree}}{\vartwo}{\varthree}&\\
\linunf{}_{\lssym}&&\diagarrow&&\linunf{}_{\lssym}\\
\esub{\esub{\vartwo}{\varone}{\vartwo}}{\vartwo}{\varthree}&\tohls&\esub{\esub{\varthree}{\varone}{\vartwo}}{\vartwo}{\varthree}& \tonhls&\esub{\esub{\varthree}{\varone}{\varthree}}{\vartwo}{\varthree}
\end{array}$$
$$\begin{array}{cccccccc}
\esub{\esub{\varone}{\varone}{\vartwo\ \vartwo}}{\vartwo}{\varthree}&&\tonhls&&\esub{\esub{\varone}{\varone}{\vartwo\ \varthree}}{\vartwo}{\varthree}&\\
\linunf{}_{\lssym}&&\diagarrow&&\linunf{}_{\lssym}\\
\esub{\esub{(\vartwo\ \vartwo)}{\varone}{\vartwo\ \vartwo}}{\vartwo}{\varthree}&\tonhls&\esub{\esub{(\vartwo\ \varthree)}{\varone}{\vartwo\ \vartwo}}{\vartwo}{\varthree}& \tonhls&\esub{\esub{(\vartwo\ \varthree)}{\varone}{\vartwo\ \varthree}}{\vartwo}{\varthree}
\end{array}$$}
%
%
%
%
%
%
%
%
%
%
2) There is an abstract lemma (see \cite{ExtendedVersion}) which says that if $\totwo\toone\subseteq \toone^+\totwo^*$ and $\toone$ is strongly normalizing then $(\toone\cup\totwo)^*\subseteq \toone^*\totwo^*$. Now, taking $\toone:=\tohls$, $\totwo:=\tonhls$ and since $\tohls$ is strongly normalising we get $\tos^*=(\tohls\cup\tonhls)^*\subseteq \tohls^*\tonhls^*$.
\end{proof}
}

We know that a $\toh$ step is simulated by a sequence of the form $\tohlb\tohls^*\tonhls^*\subseteq\tohl^*\tonhls^*$. Consider two (or more) $\toh$-steps. They are simulated by a sequence of the form $\tohl^*\tonhls^*\tohl^*\tonhls^*$, while we would like to obtain $\tohl^*\tonhls^*$. What we need to do is to prove that a sequence of the form $\tonhls^*\tohlb$ can always be reorganized as a sequence $\tohlb\tonhls^*$:
\begin{lemma}
\label{l:nhls-and-hlb}
The following inclusions hold:
\begin{varenumerate}
\item \label{p:nhls-and-hlb-one}$\tonhls\tohlb \subseteq \tohlb\tonhls$.
\item \label{p:nhls-and-hlb-two}$\tonhls^*\tohlb \subseteq \tohlb\tonhls^*$.
\end{varenumerate}
\end{lemma}

\condinc{
\begin{proof}
\begin{varenumerate}
\item By induction on $\tmone\tonhls\tmtwo$, see \cite{ARTATR} for the technical details. The idea is that $\tonhls$ cannot create, duplicate nor alter the head nature of $\tohlb$ redexes, and so the $\tohlb$-step can be preponed. Conversely, $\tohlb$-steps cannot erase, duplicate nor alter the non-head nature of $\tonhls$ redexes, therefore the two steps commute.
\item Let $t\tonhls^k\tohlb u$. The proof is by induction on $k$ using point 1, and it is a standard diagram chasing.
\end{varenumerate}
\end{proof}
}
{
\begin{proof}
1) By induction on $\tmone\tonhls\tmtwo$. The idea is that $\tonhls$ cannot create, duplicate nor alter the head nature of $\tohlb$ redexes, therefore $\tohlb$-step can be preponed. Conversely, $\tohlb$-steps cannot erase, duplicate nor alter the non-head nature of $\tonhls$ redexes, so the two steps commute.\\
2) Let $t\tonhls^k\tohlb u$. By induction on $k$ using point 1 and a standard diagram chasing.
\end{proof}
}
The next lemma is the last brick for the projection.

\begin{lemma}
\label{l:h-and-hlb-nfs}
If $\totunf{\tmone}\toh \tmtwo$ then there exists $\tmthree$ s.t. $\linunf{\tmone}\tohlb \tmthree\tos^*\tmtwo$.
\end{lemma}

\begin{proof}
By lemma \ref{l:head-sim} we get $\totunf{\tmone}\tohl\tos^* \tmtwo$. By lemma \ref{l:j-h-and-nh} $\tmone\tos^*\totunf{\tmone}$ factors as $\tmone\tohls^*\linunf{\tmone}\tonhls^*\totunf{\tmone}$, and so we get $\linunf{\tmone}\tonhls^*\tohl\tos^*\tmtwo$. By lemma \ref{l:nhls-and-hlb}.\ref{p:nhls-and-hlb-two} we get $\linunf{\tmone}\tohl\tonhls^*\tos^*\tmtwo$, \ie\ $\linunf{\tmone}\tohl\tos^*\tmtwo$.
\end{proof}
We can then conclude with the result which gives the name to this section:
\begin{lemma}[Projection of $\toh$ on $\tohl$]
\label{l:h-to-hl}
Let $\tmone$ be a $\l$-term. If $\tmone\toh^k \tmtwo$ then there exists a reduction $\redone$ s.t. $\redone: \tmone\tohl^* \tmthree$, with $\tmthree\tos^*\tmtwo$ and $\sizeb{\redone}=k$.
\end{lemma}

\begin{proof}
By induction on $k$. The case $k=0$ is trivial, so let $k>0$  and $\tmone\toh^{k-1} \tmfour\toh\tmtwo$. By \ih\ there exists a reduction $\redtwo$ s.t. $\redtwo:\tmone\tohl^* \tmfive$, $\tmfive\tos^*\tmfour$ and $\sizeb{\redtwo}=k-1$. Since $\tmfour$ is a $\l$-term we have that $\totunf{\tmfive}=\tmfour$ and $\totunf{\tmfive}\toh\tmtwo$. By lemma \ref{l:h-and-hlb-nfs} there exists $\tmthree$ s.t. $\linunf{\tmfive}\tohlb \tmthree\tos^*\tmtwo$. Moreover, $\tmfive\tohls^*\linunf{\tmfive}$; call $\redthree$ this reduction. Let $\redone$ be the reduction obtained by concatenating $\redtwo:\tmone\tohl^* \tmfive$, $\redthree:\tmfive\tohls^*\linunf{\tmfive}$ and $\linunf{\tmfive}\tohlb \tmthree$. We have $\sizeb{\redone}=\sizeb{\redtwo}+1=k$ and $\tmthree\tos^*\tmtwo$, and so we conclude.
\end{proof}

\subsection{Quadratic Relation}\label{ss:quadratic}
The last two sections proved that head reduction can be seen as head linear reduction followed by unfolding, 
and that the number of head steps is exactly the number of multiplicative steps in the corresponding head linear 
reduction. To conclude that head and head linear reduction are polynomially related we need to show 
that the number of exponential steps in a linear head reduction $\rho$ is a polynomial of the number of 
multiplicative steps in $\rho$.
 
We do it by giving a precise estimate of the maximal length of a $\tohls$ reduction from a given term. Intuition tells us 
that any reduction $\tmone\tohls^*\tmtwo$ cannot be longer than the number of explicit substitutions in $\tmone$ 
(number noted $\esnum{\tmone}$), since any substitution in $\tmone$ can act at most once on the head variable.  
However, a formal proof of this fact is not completely immediate, and requires to introduce a measure and prove some (easy) lemmas.

The idea is to statically count the length of the maximum chain of substitutions on the head, and to show that this 
number decreases at each head linear substitution step. Let us give an example. Consider the reduction:
$$
\tmone=\esub{\esub{(\varone\ \vartwo)}{\varone}{\vartwo\ \tmthree}}{\vartwo}{\tmtwo}
\tohls 
\esub{\esub{((\vartwo\ \tmthree)\ \vartwo)}{\varone}{\vartwo\ \tmthree}}{\vartwo}{\tmtwo}
\tohls
\esub{\esub{((\tmtwo\ \tmthree)\ \vartwo)}{\varone}{\vartwo\ \tmthree}}{\vartwo}{\tmtwo}
$$
It is easy to establish statically on $\tmone$ that $\esub{}{\vartwo}{\tmtwo}$ will give rise to the second $\tohls$-step, 
since $\vartwo$ is the head variable of $\vartwo\ \tmthree$, which is what is going to be substituted on the head variable of 
$\tmone$, \ie\ $\esub{}{\vartwo}{\tmtwo}$ is an \emph{hereditary} head substitution of $\tmone$. We use this idea to define 
the measure. Note that, according to our reasoning, $\esub{}{\vartwo}{\tmtwo}$ is an hereditary head substitution also for 
$\tmfour=\esub{(\varone\ \vartwo)}{\varone}{\esub{(\vartwo\ \tmthree)}{\vartwo}{\tmtwo}}$, but we get around these nested 
cases because we only have to deal with shallow terms.

\begin{definition}
\deff{Hereditary head contexts} are generated by the following grammar:
$$\hh := \econe \mid \esub{\hhp{\varone}}{\varone}{\econe}$$
The \deff{head measure} $\hhmes{\tmone}$ of a shallow term $\tmone$ is defined by induction on $\tmone$:
\[\begin{array}{lll@{\sep\sep\sep\sep\sep}llllll}
\hhmes{\varone}&=&0 &
\hhmes{\esub{\tmone}{\varone}{\tmtwo}}&=&\hhmes{\tmone} & \mbox{ if }\tmone\neq\hhp{\varone}\\

\hhmes{\l \varone.\tmone}&=&\hhmes{\tmone}&\hhmes{\esub{\tmone}{\varone}{\tmtwo}}&=&\hhmes{\tmone} +1 & \mbox{otherwise}\\

\hhmes{\tmone\ \tmtwo}&=&\hhmes{\tmone}
\end{array}\]
\end{definition}
Please notice that $\hhmes{\tmone}=0$ for any $\l$-term $\tmone$. 
\condinc{
We need a lemma about $\hh$-contexts.
\begin{lemma}
\label{l:hh-ctx}
Let $\tmone$ be a shallow term. 
\begin{varenumerate}
 \item \label{p:hhzero-imp-h}If $\hhmes{\tmone}=0$ and $\tmone=\hhp{\varone}$ then the context $\hhp{\cdot}$ is an head context.
  \item \label{p:tohls-ctx}If $\tmone$ is shallow, $\tmone\tohls \tmtwo$, $\tmone=\hhp{\varone}$ and $\hhp{\cdot}$ does not capture $\varone$, then there exists $\hh_0$ s.t. $\tmtwo=\hh_0[\varone]$ and $\hh_0$ does not capture $\varone$.
\end{varenumerate}
\end{lemma}
\condinc{
\begin{proof}
\begin{varenumerate}
 \item By induction on $\tmone$. If $\tmone=\esub{\tmthree}{\vartwo}{\tmtwo}$ the hypothesis $\hhmes{\tmone}=0$ implies that $\tmthree\neq\hhp{\vartwo}$ and $\hhmes{\tmthree}=0$. If $\tmone=\hhp{\varone}$ for some $\varone$ then $\hhp{\cdot}=\esub{\hh_0[\cdot]}{\vartwo}{\tmtwo}$ for some hereditary head context $\hh_0$, and $\tmthree=\hh_0[\varone]$. By \ih\ $\hh_0$ is an head context, and so is $\hh$. All other cases are straightforward.

\item By induction on $\tmone\tohls \tmtwo$. There are two interesting cases:
\begin{varitemize}
\item \textit{Base case}: $\tmone=\esub{\econep{\vartwo}}{\vartwo}{\tmthree}\tohls \esub{\econep{\tmthree}}{\vartwo}{\tmthree}=\tmtwo$. Then $\tmthree=\econe_0[\varone]$ and $\tmtwo=\esub{\econep{\econe_0[\varone]}}{\vartwo}{\tmthree}$. We conclude by defining $\hh_0:=\esub{\econep{\econe_0}}{\vartwo}{\tmthree}$.

\item \textit{Substitution inductive case}: $\tmone=\esub{\tmthree}{\vartwo}{\tmfive}\tohls \esub{\tmfour}{\vartwo}{\tmfive}=\tmtwo$ because $\tmthree\tohls\tmfour$. There are two subcases:
\begin{varitemize}
\item $\tmone=\esub{\hhp{\varone}}{\vartwo}{\tmfive}$ (and $\varone\neq \vartwo$). Then by \ih\ exists $\hh_0'$ s. t. $\tmfour=\hh_0'[\varone]$ and so we conclude taking $\hh_0:=\esub{\hh_0'}{\vartwo}{\tmfive}$.
\item $\tmone=\esub{\hhp{\vartwo}}{\vartwo}{\tmfive}$ and $\tmfive=\econep{\varone}$. Then by \ih\ exists $\hh_0'$ s. t. $\tmfour=\hh_0'[\vartwo]$ and so we conclude taking $\hh_0:=\esub{\hh_0'[\vartwo]}{\vartwo}{\econe}$.
\end{varitemize}
\end{varitemize}
\end{varenumerate}
This concludes the proof.
\end{proof}
}
{
\begin{proof}
1) By induction on $\tmone$. 2) By induction on $\tmone\tohls \tmtwo$.
\end{proof}
}
Next lemma proves that $\hhmes{\tmone}$ correctly captures the number of $\tohls$-reductions from $\tmone$, what is then compactly expressed by the successive corollary.
}{
Next lemma proves that $\hhmes{\tmone}$ correctly captures the number of $\tohls$-reductions from $\tmone$, what is then compactly expressed by the successive corollary. The detailed proof of the lemma is in \cite{ExtendedVersion}.
}
\begin{lemma}[$\hhmes{\cdot}$ decreases with $\tohls$]
\label{l:hhmes-decr}
Let $\tmone$ be a shallow term.
\begin{varenumerate}
\item \label{p:hhmes-zero}$\tmone$ is a $\tohls$-normal form iff $\hhmes{\tmone}=0$.
\item \label{p:hhmes-decr-one}$\tmone\tohls \tmtwo$ implies $\hhmes{\tmone}=\hhmes{\tmtwo}+1$.
\item \label{p:hhmes-decr-two}$\hhmes{\tmone}>0$ implies that $\tmone$ is not $\tohls$-normal.
\end{varenumerate}
\end{lemma}

\condinc{
\begin{proof}
\begin{varenumerate}
\item $\Rightarrow$) By induction on $t$. The only interesting case is if $\tmone=\esub{\tmtwo}{\varone}{\tmthree}$. Then $\tmtwo$ is a $\tohls$-normal form and so by \ih\ we get $\hhmes{\tmtwo}=0$. By hypothesis $\tmtwo\neq\econep{\varone}$, otherwise $\tmone$ would not be $\tohls$-normal. By Lemma \ref{l:hh-ctx}.\ref{p:hhzero-imp-h} if $\tmtwo=\hhp{\vartwo}$ then $\hh$ is an head context. Since for any term $\tmfive$ there is exactly one head context $\econe_0$ s.t. $\tmfive=\econe_0[\tmfour]$ and $\tmfour$ is a variable, we get that $\hh=\econe$ and $\tmtwo\neq\hhp{\varone}$. By definition of $\hhmes{\cdot}$ we get $\hhmes{\tmone}=\hhmes{\tmtwo}=0$.\\
$\Leftarrow$) By induction on $t$. The only interesting case is if $\tmone=\esub{\tmtwo}{\varone}{\tmthree}$. The hypothesis implies that $\hhmes{\tmtwo}=0$ and $\tmtwo\neq\hhp{\varone}$. In particular, $\tmtwo\neq\econep{\varone}$ and $\esub{\cdot}{\varone}{\tmthree}$ does not give an $\tohls$-redex. By \ih\ $\tmtwo$ is a $\tohls$-normal form. Then $\tmone$ is a $\tohls$-normal form.

\item 
By induction on $\tmone\tohls \tmtwo$. Cases:
\begin{varitemize}
\item $\tmone=\esub{\econep{\varone}}{\varone}{\tmthree}\tohls \esub{\econep{\tmthree}}{\varone}{\tmthree}=\tmtwo$. Since $\varone$ is free in $\econep{\varone}$ we get that $\econep{\varone}$ is a $\tohls$-normal form, and by Point \ref{p:hhmes-zero} $\hhmes{\econep{\varone}}=0$. It follows that $\hhmes{\tmone}=1$. The hypothesis that $\tmone$ is shallow implies that $\tmthree$ is a $\l$-term and the hypothesis $\esub{\econep{\varone}}{\varone}{\tmthree}\tohls \esub{\econep{\tmthree}}{\varone}{\tmthree}$ implies that $\econe$ does not capture any free variable of $\tmthree$. Then $\econep{\tmthree}$ is a $\tohls$-normal form and by Point \ref{p:hhmes-zero} $\hhmes{\econep{\tmthree}}=0$. By lemma \ref{l:hh-ctx}.\ref{p:hhzero-imp-h} if $\econep{\tmthree}$ has the form $\hhp{\vartwo}$ for some $\vartwo$ then $\hhp{\cdot}$ is an  head context, and so $\vartwo$ is a variable of $\tmthree$. Since $\varone\notin\fv{\tmthree}$ we get $\varone\neq\vartwo$ and $\econep{\tmthree}$ has not the form $\hhp{\varone}$. Hence $\hhmes{\esub{\econep{\tmthree}}{\varone}{\tmthree}}=\hhmes{\tmtwo}=0$, and $\hhmes{\tmone}=\hhmes{\tmtwo}+1$.

\item $\tmone= \l \varone. \tmthree\tohls \l \varone. \tmfour=\tmtwo$. Using the \ih.
\item $\tmone= \tmthree\ \tmfive\tohls \tmfour\ \tmfive=\tmtwo$. Using the \ih.

\item $\tmone=\esub{\tmthree}{\varone}{\tmfive}\tohls \esub{\tmfour}{\varone}{\tmfive}=\tmtwo$. By \ih\ $\hhmes{\tmthree}=\hhmes{\tmfour}+1$. The measures $\hhmes{\tmone}$ and $\hhmes{\tmtwo}$ are given by $\hhmes{\tmthree}$ and $\hhmes{\tmfour}$, respectively, plus the eventual contribution of $\esub{}{\varone}{\tmfive}$. By lemma \ref{l:hh-ctx}.\ref{p:tohls-ctx} either $\esub{}{\varone}{\tmfive}$ contributes to both $\hhmes{\tmone}$ and $\hhmes{\tmtwo}$ or it contributes to none. In both cases we get $\hhmes{\tmone}=\hhmes{\tmtwo}+1$.
\end{varitemize}

\item By induction on $\tmone$. The base case $\tmone=\varone$ is trivial, and the cases $\tmone=\l \varone.\tmtwo$ and $\tmone=\tmtwo\ \tmthree$ follows from the \ih. The only interesting case is when
 $\tmone=\esub{\tmtwo}{\varone}{\tmthree}$. If $\hhmes{\tmtwo}>0$ then we conclude using the \ih. Otherwise, it must be that $\tmtwo=\hhp{\varone}$. By Lemma \ref{l:hh-ctx}.\ref{p:hhzero-imp-h} $\hh$ is an head context and so $\tmone=\esub{\econep{\varone}}{\varone}{\tmthree}\tohls \esub{\econep{\tmthree}}{\varone}{\tmthree}$ and we conclude.
\end{varenumerate}
\end{proof}
}
{
\begin{proof}
1) By induction on $\tmone$. 2) By induction on $\tmone\tohls \tmtwo$. 3) By induction on $\tmone$. 
\end{proof}}
Summing up, we get:

\begin{corollary}[exact bound to $\tohls$-sequences]
\label{coro:bound-tohls}
$\tmone\tohls^n \linunf{\tmone}$ iff $n=\hhmes{\tmone}$.
\end{corollary}

\condinc{
\begin{proof}
By induction on $n$. For $n=0$ the statement is given by Lemma \ref{l:hhmes-decr}.\ref{p:hhmes-zero}. Then let $n>0$.\\
Direction $\Rightarrow$: if $\tmone\tohls\tmtwo\tohls^{n-1} \linunf{\tmone}$ then by \ih\ $\hhmes{\tmtwo}=n-1$. By Lemma \ref{l:hhmes-decr}.\ref{p:hhmes-decr-one} we get $\hhmes{\tmone}=n$.\\
Direction $\Leftarrow$: by Lemma \ref{l:hhmes-decr}.\ref{p:hhmes-decr-two} $\tmone$ is not $\tohls$-normal and so $\tmone\tohls \tmtwo$. By Lemma \ref{l:hhmes-decr}.\ref{p:hhmes-decr-one} we get $\hhmes{\tmtwo}=n-1$ and by \ih\ $\tmtwo\tohls^{n-1} \linunf{\tmtwo}=\linunf{\tmone}$, and so $\tmone\tohls^n \linunf{\tmone}$.
\end{proof}
}
{
\begin{proof}
By induction on $n$. For $n=0$ see Lemma \ref{l:hhmes-decr}.\ref{p:hhmes-zero}, for $n>0$ it follows from \ref{l:hhmes-decr}.\ref{p:hhmes-decr-one}-\ref{p:hhmes-decr-two}.
\end{proof}
}
Now, we are ready to prove the quadratic relation. The following lemma is the key point for the combinatorial analysis. It shows that if the initial term $\tmone$ of a reduction $\redone:\tmone\tohl^n\tmtwo$ is a $\l$-term then $\hhmes{\tmtwo}$ is bounded by the number of $\tohlb$ steps in $\redone$ (noted $\sizeb{\redone}$).

\begin{lemma}
\label{l:enum-sizeb}
Let $\tmone\in\lterms$. If $\redone:\tmone\tohl^n\tmtwo$ then $\hhmes{\tmtwo}\leq \esnum{\tmtwo}=\sizeb{\redone}$.
\end{lemma}

\condinc{
\begin{proof}
Note that by definition of $\hhmes{\cdot}$ we get $\hhmes{\tmtwo}\leq \esnum{\tmtwo}$ for any term $\tmtwo$. So we only need prove that $\esnum{\tmtwo}=\sizeb{\redone}$. 
By induction on $k=\sizeb{\redone}$. If $k=0$ then $\redone$ is empty, because $\tmone$ is a $\l$-term and so it is $\tohls$-normal. Then $\tmone=\tmtwo$ and $\esnum{\tmtwo}=0$. If $k>0$ then $\redone=\redtwo;\tohlb;\tohls^m$  for some $m$ and some reduction $\redtwo$. Let $\tmthree$ be the end term of $\redtwo$ and $\tmfour$ the term s.t. $\tmthree\tohlb\tmfour\tohls^*\tmtwo$. By 
\ih\ $\esnum{\tmthree}= \sizeb{\redtwo}=\sizeb{\redone}-1$. Now, $\esnum{\tmfour}=\esnum{\tmthree}+1=\sizeb{\redone}$, because each $\tohlb$-step creates an explicit substitution. It is easy to see that $\tohls$-steps do not change the number of substitutions in a term (\ie\ $\esnum{\tmfour}$): by lemma \ref{l:box-subterm} we get that any box subterm of $\tmfour$ is a box-subterm of $\tmone$, and since $\tmone$ is a $\l$-term, the duplication performed by a $\tohls$-step does not increase the number of explicit substitutions. Therefore, $\esnum{\tmtwo}=\esnum{\tmfour}=\sizeb{\redone}$.
\end{proof}
}
{
\begin{proof}
Note that by definition of $\hhmes{\cdot}$ we get $\hhmes{\tmtwo}\leq \esnum{\tmtwo}$ for any term $\tmtwo$. So we only need prove that $\esnum{\tmtwo}=\sizeb{\redone}$. 
By induction on $k=\sizeb{\redone}$. If $k=0$ then $\redone$ is empty, because $\tmone$ is a $\l$-term and so it is $\tohls$-normal. Then $\tmone=\tmtwo$ and $\esnum{\tmtwo}=0$. If $k>0$ then $\redone=\redtwo;\tohlb;\tohls^m$  for some $m$ and some reduction $\redtwo$. Let $\tmthree$ be the end term of $\redtwo$ and $\tmfour$ the term s.t. $\tmthree\tohlb\tmfour\tohls^*\tmtwo$. By 
\ih\ $\esnum{\tmthree}= \sizeb{\redtwo}=\sizeb{\redone}-1$. Now, $\esnum{\tmfour}=\esnum{\tmthree}+1=\sizeb{\redone}$, because each $\tohlb$-step creates an explicit substitution. By lemma \ref{l:box-subterm} we get that any box subterm of $\tmfour$ is a box-subterm of $\tmone$, and since $\tmone$ is a $\l$-term, the duplication performed by a $\tohls$-step does not increase the number of explicit substitutions. Therefore, $\esnum{\tmtwo}=\esnum{\tmfour}=\sizeb{\redone}$.
\end{proof}
}
We finally get:

\begin{theorem}\label{th:quadraticbound}
Let $\tmone\in\lterms$. If $\redone:\tmone\tohl^n\tmtwo$ then $n=O(\sizeb{\redone}^2)$.
\end{theorem}
\begin{proof}
There exists $k\in\nat$ s.t. $\redone=\redtwo_1;\redthree_1;\ldots;\redtwo_k;\redthree_k$, where $\redtwo_i$ is a non-empty 
$\tohlb$-reduction and $\redthree_i$ is a $\tohls$-reduction for $i\in\set{1,\ldots,k}$ and it is non-empty for $i\in\set{1,\ldots,k-1}$. \\
Let $\tmthree_1,\ldots, \tmthree_k$ be the end terms of $\tau_1, \ldots, \tau_k$, respectively. 
By Corollary \ref{coro:bound-tohls} $|\redthree_j|\leq\hhmes{\tmthree_j}$ and by Lemma \ref{l:enum-sizeb} $\hhmes{\tmthree_j}\leq\sum_{i\in\set{1,\ldots,j}}|\redtwo_i|$. Now $\sizeb{\redone}=\sum_{i\in\set{1,\ldots,k}}|\redtwo_i|$ bounds every $\hhmes{\tmthree_j}$, hence:
$$\sum_{i\in\set{1,\ldots,k}}|\redthree_i|\leq \sum_{i\in\set{1,\ldots,k}}\hhmes{\tmthree_i}\leq k\cdot\sizeb{\redone}$$
But $k$ is bounded by $\sizeb{\redone}$ too, thus $\sum_{i\in\set{1,\ldots,k}}|\redthree_i|\leq\sizeb{\redone}^2$ and $n\leq \sizeb{\redone}^2 + \sizeb{\redone}= O(\sizeb{\redone}^2)$.
\end{proof}
Putting together the results from the whole of Section~\ref{s:relation}, we get:
\begin{corollary}[Invariance, Part I]\label{th:poly}
There is a polynomial time algorithm that, given $\tmone\in\lterms$, computes
a term $\tmtwo$ such that $\totunf{\tmtwo}=\tmthree$ if $\tmone$ has $\toh$-normal form 
$\tmthree$ and diverges if $\tmtwo$ has no $\toh$-normal form. Moreover, the algorithm 
works in polynomial time on the derivation complexity of the input term.
\end{corollary}
One may now wonder why a result like Corollary~\ref{th:poly} cannot be generalized
to, e.g., leftmost-outermost reduction, which is a normalizing strategy. Actually, 
linear explicit substitutions \emph{can} be endowed with a notion of reduction by levels
capable of simulating the leftmost-outermost strategy in the same sense as 
linear head-reduction simulates head-reduction here. And, noticeably, the subterm
property \emph{continues} to hold. What is not true anymore, however, is the quadratic
bound we have proved in this section: in the leftmost-outermost strategy, one needs to perform \emph{too many} 
substitutions \cben{which do not lead to any multiplicative redex}{not related to any $\beta$-redex}. If one wants to generalize
Corollary~\ref{th:poly}, in other words, one needs to \ben{further} optimize the substitution process.
But this is outside the scope of this paper.

\section{$\Lambda_{[\cdot]}$ as an Acceptable Encoding of $\lambda$-terms}\label{s:unfolding}
The results of the last two sections can be summarized as follows: linear explicit substitutions
provide both a compact representation for $\lambda$-terms and \cben{a way to implement}{an implementation of} $\beta$-reduction
in such a way that the overhead due to substitutions remains under control and is polynomial in the
unitary cost of the $\lambda$-term we start from. But one may wonder whether \cben{linear explicit}{explicit}
substitutions are nothing more than a way to \emph{hide} the complexity of the problem under the carpet
of compactness: what if we want to get the normal form in the \emph{usual}, \emph{explicit} form?
Counterexamples from Section~\ref{s:informal}, read through the lenses of Theorem~\ref{th:poly}
tell us that that this is \emph{indeed} the case: there are families of $\lambda$-terms with
polynomial unitary cost but whose normal form intrinsically requires exponential 
time to be produced.

In this section, we show that this phenomenon is due to the $\lambda$-calculus being a very inefficient
way to represent $\lambda$-terms: even if computing the unfolding of a term $\tmone\in\Lambda_{[\cdot]}$ 
takes exponential time, \emph{comparing} the unfoldings of two terms $\tmone,\tmtwo\in\Lambda_{[\cdot]}$
for equality can be done in polynomial time. This way, linear explicit substitutions are proved to
be a succint, acceptable encoding of $\lambda$-terms in the sense of Papadimitriou~\cite{Papadimitriou}.
The algorithm we are going to present is based on dynamic programming:
it compares all the \emph{relative unfoldings} of subterms of two terms $\tmone$ and $\tmtwo$ as above,
without really computing those unfoldings. Some complications arise due to the underlying notion of
equality for $\lambda$-terms, namely $\alpha$-equivalence, which is coarser than syntactical
equivalence. But what is the relative unfolding of a term?
\begin{definition}[Relative Unfoldings]
The unfolding $\parunf{\tmone}{\ctxone}$ of $\tmone$ relative to context $\ctxone$ is defined by induction on 
$\ctxone$: 
\begin{align*}
\parunf{\tmone}{\hole}&:=\totunf{t};&\parunf{\tmone}{\tmtwo\ \ctxone}&:=\parunf{\tmone}{\ctxone};
  &\parunf{\tmone}{\esub{\tmtwo}{\varone}{\ctxone}}&:=\parunf{\tmone}{\ctxone};\\
\parunf{\tmone}{\ctxone\ \tmtwo}&:=\parunf{\tmone}{\ctxone};&\parunf{\tmone}{\l \varone. \ctxone} &:=\parunf{\tmone}{\ctxone};
  &\parunf{\tmone}{\esub{\ctxone}{\varone}{\tmtwo}}&:=\isub{\parunf{\tmone}{\ctxone}}{\varone}{\totunf{\tmtwo}}.
\end{align*}
\end{definition}
Constraining sets allow to give sensible judgments about the equivalence of 
terms even when their free variable differ:
\begin{definition}[Constraining Sets and Coherence]
A \emph{constraining set} $\conone$ is a set of pairs $(\varone,\vartwo)$ of variable names. 
Two constraining sets $\conone$ and $\contwo$ are \emph{coherent} (noted $\conone\coh\contwo$) if:
\begin{varitemize}
\item  $(\varone,\vartwo)\in\conone$ and $(\varone,\varthree)\in\contwo$ imply $\vartwo=\varthree$; 
\item  $(\vartwo,\varone)\in\conone$ and $(\varthree,\varone)\in\contwo$ imply $\vartwo=\varthree$.
\end{varitemize}
Moreover, $\conone$ is \emph{auto-coherent} if $\conone\coh\conone$. Observe that $\coh$ is
not reflexive and that a constraining set is auto-coherent iff it is the graph of a bijection. 
\end{definition}
\newcommand{\subvar}{\mathcal{S}}
\newcommand{\ujlt}[2]{\sqsubset_{#1,#2}}
\newcommand{\ujgt}[2]{\sqsupset_{#1,#2}}
\newcommand{\pone}{P}
\newcommand{\ptwo}{Q}
\newcommand{\pthree}{R}
\newcommand{\pfour}{S}
The algorithm tests pairs $(\btmone,\btmtwo)$ of terms. We assume them \emph{preprocessed} as follows: the spaces of 
substituted, abstracted and free names of $\btmone$ and $\btmtwo$ are all pairwise disjoint and neither $\btmone$ nor $\btmtwo$ 
contain any subterm in the form $\esub{\btmthree}{\varone}{\btmfour}$, where $\varone\not\in\fv{\btmthree}$\footnote{Any term can be turned in this form in polynomial time, by $\tow$-normalizing it.}. We also note 
$\subvar$ the set of substituted variables of both terms. The whole algorithm is built around the notion of an unfolding judgment:
\begin{definition}[Unfolding Judgments]
Let $\pairone=(\btmone,\btmtwo)$ be a preprocessed pair. An unfolding judgement is a triple $(\tmone,\ctxone), \valone, (\tmtwo,\ctxtwo)$, where 
$\valone$ is either $\bot$ or a constraining set $\conone$, also noted $(\tmone,\ctxone)\unf{\valone}(\tmtwo,\ctxtwo)$. The rules for deriving unfolding 
judgments are in Figure \ref{f:unfolding-rules}. An operation $\valop$ on values is used, which is defined as follows:
\begin{varenumerate}
\item $\valone\valop\valtwo=\valone\cup\valtwo$	if $\valone\neq\bot\neq\valtwo$ and $\valone\coh\valtwo$;
\item $\valone\valop\valtwo=\bot$	if $\valone\neq\bot\neq\valtwo$ and $\valone\not\coh\valtwo$;
\item $\valone\valop\valtwo=\bot$ if $\valone=\bot$ or $\valtwo=\bot$.
\end{varenumerate}
The rules in Figure \ref{f:unfolding-rules} induces a binary relation $\ujlt{\btmone}{\btmtwo}$ on the space of pairs (of pairs)
in the form $((\tmone,\ctxone),(\tmtwo,\ctxtwo))$ such that $\ctxone[\tmone]=\btmone$ and $\ctxtwo[\tmtwo]=\btmtwo$:
$(\pone,\ptwo)\ujlt{\btmone}{\btmtwo}(\pthree,\pfour)$ if knowing
$\valone$ such that $\pone\unf{\valone}\ptwo$ is necessary to compute $\valtwo$
such that $\pthree\unf{\valtwo}\pfour$.
\end{definition}
\begin{figure}[t]
\input{unfolding-rules}
\caption{\label{f:unfolding-rules} Unfolding rules}
\end{figure}
\newcommand{\blank}{\#}
\newcommand{\yes}{\mathtt{yes}}
\newcommand{\no}{\mathtt{no}}
\condinc{
\begin{lemma}
The relation $\ujlt{\btmone}{\btmtwo}^*$ is a partial order, while $\ujlt{\btmone}{\btmtwo}^+$ is
a strict order.
\end{lemma}
\begin{proof}
Let $\lessdot$ the strict order on $(\Nset\times\Nset)\times(\Nset\times\Nset)$ defined
as the product order of the standard lexicographic order on $\Nset\times\Nset$. Let $||\ctxone||$ be the
number of substitutions into which $\hole$ is embedded inside $\ctxone$. Observe that 
reflexivity and transitivity of $\ujlt{\btmone}{\btmtwo}^*$ hold by definition. About antisimmetry, just
observe that if $((\tmone,\ctxone),(\tmtwo,\ctxtwo))\ujlt{\btmone}{\btmtwo}((\tmthree,\ctxthree),(\tmfour,\ctxfour))$,
then $((||\ctxthree||,|\ctxthree|),(||\ctxfour||,|\ctxfour|))\lessdot((||\ctxone||,|\ctxone|),(||\ctxtwo||,|\ctxtwo|))$.
\end{proof}
\begin{lemma}
For every $\pone$ for $\btmone$ and for every $\ptwo$ for
$\btmtwo$ there is exactly one $\valone$ such that $\pone\unf{\valone}\ptwo$.
\end{lemma}
\begin{proof}
Let $\pone=(\tmone,\ctxone)$ and $\ptwo=(\tmtwo,\ctxtwo)$. We proceed
by induction on the relation $\ujgt{\btmone}{\btmtwo}$, which is a strict
order on a finite set, thus a well order. Let us distinguish some cases depending
on the form of $\tmone$ and $\tmtwo$:
\begin{varitemize}
\item
  If both $\tmone$ and $\tmtwo$ are variables not in $\subvar$, then
  we can apply rule $\rvar$.
\item
  If both $\tmone$ and $\tmtwo$ are abstractions, then we can apply 
  the inductive hypothesis to the pairs obtained by taking the bodies
  of these two abstractions and apply the inductive hypothesis. Depending
  on the (unique!) outcome, we can apply exactly one
  between $\rabsoc$, $\rabsweak$ and $\rabsinc$. 
\item
  If both $\tmone$ and $\tmtwo$ are applications, then we can
  apply the inductive hypothesis to the pairs obtained by taking
  the immediate subterm, and conclude observing the shape of rule 
  $\rapp$.
\item
  If any of $\tmone$ and $\tmtwo$ is a substitution, then we can
  proceed as usual and then apply either rule $\rsubl$ or $\rsubr$.
\item
  If any of $\tmone$ and $\tmtwo$ is a substituted variable, then
  we can again apply the inductive hypothesis, and the rule
  $\runfl$ or $\runfr$.
\item
  In all the other cases, we can easily conclude by observing
  the shape of the six rules $\rerrla$, $\rerral$. $\rerrlv$, $\rerrvl$,
  $\rerrva$, $\rerrav$.
\end{varitemize}
This concludes the proof.
\end{proof}
\begin{lemma}\label{lemma:cardval}
If $\pone$ and $\ptwo$ are pairs
for $\btmone$ and $\btmtwo$ (respectively)
and $\pone\unf{\valone}\ptwo$, then
$|\valone|\leq|\btmone|\cdot|\btmtwo|$.
\end{lemma}
\begin{proof}
By induction on the structure of the proof
that $\pone\unf{\conone}\ptwo$, it is easy
to prove that if $(\varone,\vartwo)\in\conone$,
then $\varone$ is free in $\btmone$ and $\vartwo$
is free in $\btmtwo$.
\end{proof}
\begin{definition}[Unfolding Matrix]
Let $\pairone=(\btmone,\btmtwo)$ be a preprocessed pair. An unfolding matrix $\matr$ for 
$\pairone$ is bidimensional array with the following form:
\begin{varenumerate}
\item its rows are indexed by the pairs $(\tmone,\ctx{\ctxone}{\cdot})$ such that $\ctx{\ctxone}{\tmone}=\btmone$;
\item its columns are indexed by the pairs $(\tmone,\ctx{\ctxone}{\cdot})$ such that $\ctx{\ctxone}{\tmone}=\btmtwo$;
\item the values of the matrix are either (possibly empty) constraining sets or $\bot$ or $\blank$.
\end{varenumerate}
\end{definition}}
{
Unfolding matrices are the data structure on which our algorithm works. 
An \emph{unfolding matrix} for a preprocessed pair $(\btmone,\btmtwo)$ is a bidimensional array
whose rows are indexed by pairs $(\tmone,\ctx{\ctxone}{\cdot})$ such that $\ctx{\ctxone}{\tmone}=\btmone$
and whose columns are indexed by pairs $(\tmone,\ctx{\ctxone}{\cdot})$ such that $\ctx{\ctxone}{\tmone}=\btmtwo$. 
Each element of the unfolding matrix can be either a (possibly empty) constraining sets or $\bot$ or $\blank$.
}
Basically, the unfolding checking algorithm simply proceeds by filling an unfolding matrix with
the correct values, following the rules in Figure~\ref{f:unfolding-rules} and starting from an
unfolding matrix filled with $\blank$:
\begin{definition}[Unfolding Checking Algorithm]
Let $\pairone=(\btmone,\btmtwo)$ be a preprocessed pair. We define the following algorithm, which
which be referred to as the \emph{unfolding checking algorithm}:
\begin{varenumerate}
\item Initialize all entries in $\matr$ to $\blank$;
\item\label{s:iter} If $\matr$ has no entries filled with $\blank$, then go to step \ref{s:final};
\item\label{s:choose} Choose $(\tmone, \ctxone)$ and $(\tmtwo, \ctxtwo)$ such that $\matr[(\tmone, \ctxone)][(\tmtwo, \ctxtwo)]=\blank$, 
  and such that $\matr[\pone_1][\ptwo_1],\ldots,\matr[\pone_n][\ptwo_n]$
  are all different from $\blank$, where $(\pone_1,\ptwo_1),\ldots,(\pone_n,\ptwo_n)$ are the immediate predecessors
  of $((\tmone, \ctxone),(\tmtwo, \ctxtwo))$ in $\ujlt{\btmone}{\btmtwo}$;
\item\label{s:compute} Compute $\valone$ such that $(\tmone,\ctxone)\unf{\valone}(\tmtwo,\ctxtwo)$;
\item $\matr[(\tmone,\ctxone)][(\tmtwo,\ctxtwo)]\leftarrow\valone$;
\item  Go to step \ref{s:iter};
\item\label{s:final} Return $\yes$ if $\matr[(\btmone,\hole)][(\btmtwo,\hole)]$ is a constraining set which is the identity, otherwise return $\no$.
\end{varenumerate}
\end{definition}
It is now time to prove that the Unfolding Checking Algorithm is correct, i.e., that it gives correct results, if any:
\condinc{
\begin{lemma}
\label{l:parunf-ctx}
Relative unfoldings verify:
\begin{varenumerate}
\item \label{l:p:parunf-ctx-zero} If $\ctxone$ does not contain an explicit substitution for $\varone$ then $\parunf{\varone}{\ctxone}=\varone$.
\item \label{l:p:parunf-ctx-one} $\parunf{\esub{\tmone}{\varone}{\tmthree}}{\ctxone}=\parunf{\tmone}{\ctx{\ctxone}{\esub{\hole}{\varone}{\tmthree}}}$.
\item \label{l:p:parunf-ctx-two} $\parunf{(\l \varone.\tmone)}{\ctxone} = \l\varone.(\parunf{\tmone}{\ctx{\ctxone}{\l \varone.\hole}})$.
\item \label{l:p:parunf-ctx-three} $\parunf{(\tmone\ \tmthree)}{\ctxone}=\parunf{\tmone}{\ctx{\ctxone}{\hole\ \tmthree}}\ \parunf{\tmthree}{\ctx{\ctxone}{\tmone\ \hole}}$.
\item \label{l:p:parunf-ctx-four} $\parunf{\tmone}{\ctx{\ctxone}{\esub{\ctx{\ctxthree}{x}}{\varone}{\hole}}}= \parunf{\varone}{\ctx{\ctxone}{\esub{\ctxthree}{\varone}{\tmone}}}$.
\end{varenumerate}
\end{lemma}
\begin{proof}
All points are by induction on $\ctxone$. Relevant cases:
\begin{varitemize}
\item 
  If $\ctxone=\hole$, then:
  \begin{align*}
    \ref{l:p:parunf-ctx-zero}.&&\parunf{\varone}{\hole}&=\varone;\\
    \ref{l:p:parunf-ctx-one}.&&\parunf{\tmone}{\esub{\hole}{\varone}{\tmthree}}&=\parunf{\tmone}{\hole}\isubs{\varone/\totunf{\tmthree}}=
       \totunf{\tmone}\isubs{\varone/\totunf{\tmthree}}
        =\totunf{(\esub{\tmone}{\varone}{\tmthree})};\\
    \ref{l:p:parunf-ctx-two}.&&\parunf{(\l \varone.\tmone)}{\hole}&=\totunf{(\l \varone.\tmone)}= \l \varone.\totunf{\tmone}=\l\varone.(\parunf{\tmone}{\hole})
        =\l\varone.(\parunf{\tmone}{\l \varone.\hole});\\
    \ref{l:p:parunf-ctx-three}.&&\parunf{(\tmone\ \tmthree)}{\hole}&=\totunf{(\tmone\ \tmthree)}=\parunf{\tmone}{\hole}\ \parunf{\tmthree}{\hole}
       =\parunf{\tmone}{(\hole\ \tmthree)}\ \parunf{\tmthree}{(\tmone\ \hole)}\\
    \ref{l:p:parunf-ctx-four}.&&\parunf{\tmone}{\esub{\ctx{\ctxthree}{\varone}}{\varone}{\hole}}&=\tmone;\\
    &&\parunf{\varone}{\esub{\ctxthree}{\varone}{\tmone}}&=\parunf{\varone}{\ctxthree}\isubs{\varone/\tmone}=_{p.\ref{l:p:parunf-ctx-zero}} \varone\isubs{\varone/\tmone}=\tmone.
  \end{align*}
\item 
  If $\ctxone=\ctxtwo\tmtwo$, then:
  \begin{align*}
    \ref{l:p:parunf-ctx-zero}.&&\parunf{\varone}{(\ctxtwo\ \tmtwo)}=\parunf{\varone}{\ctxtwo}&=_{\ih}\varone;\\
    \ref{l:p:parunf-ctx-one}.&&\parunf{\tmone}{\ctx{\ctxtwo}{\esub{\hole}{\varone}{\tmthree}}\ \tmtwo}&=
      \parunf{\tmone}{\ctx{\ctxtwo}{\esub{\hole}{\varone}{\tmthree}}}=_{\ih}
      \parunf{\esub{\tmone}{\varone}{\tmthree}}{\ctxtwo}=
      \parunf{\esub{\tmone}{\varone}{\tmthree}}{(\ctxtwo\ \tmtwo)};\\
    \ref{l:p:parunf-ctx-two}.&&\parunf{(\l \varone.\tmone)}{(\ctxtwo\ \tmtwo)}&= \parunf{(\l \varone.\tmone)}{\ctxtwo}=_{\ih}
      \l\varone.(\parunf{\tmone}{\ctx{\ctxtwo}{\l \varone.\hole}})=
      \l\varone.(\parunf{\tmone}{\ctx{\ctxtwo}{\l \varone.\hole}\ \tmtwo});\\
    \ref{l:p:parunf-ctx-three}.&&\parunf{(\tmone\ \tmthree)}{(\ctxtwo\ \tmtwo)}&=
      \parunf{(\tmone\ \tmthree)}{\ctxtwo}=_{\ih}
      \parunf{\tmone}{\ctx{\ctxtwo}{\hole\ \tmthree}}\ \parunf{\tmthree}{\ctx{\ctxtwo}{\tmone\ \hole}}\\
      &&&=\parunf{\tmone}{(\ctx{\ctxtwo}{\hole\ \tmthree}\ \tmtwo)}\ \parunf{\tmthree}{(\ctx{\ctxtwo}{\tmone\ \hole}\ \tmtwo)};\\
    \ref{l:p:parunf-ctx-four}.&&\parunf{\tmone}{(\ctx{\ctxtwo}{\esub{\ctx{\ctxthree}{x}}{\varone}{\hole}}\ \tmtwo)}&= 
      \parunf{\tmone}{\ctx{\ctxtwo}{\esub{\ctx{\ctxthree}{x}}{\varone}{\hole}}}=_{\ih}
      \parunf{\varone}{\ctx{\ctxtwo}{\esub{\ctxthree}{\varone}{\tmone}}}\\
      &&&=\parunf{\varone}{(\ctx{\ctxtwo}{\esub{\ctxthree}{\varone}{\tmone}}\ \tmtwo)}.
  \end{align*}
\item 
  If $\ctxone=\esub{\ctxtwo}{\vartwo}{\tmtwo}$, then:
  \begin{align*}
    \ref{l:p:parunf-ctx-zero}.&&\parunf{\varone}{\esub{\ctxtwo}{\vartwo}{\tmtwo}}&=
      \parunf{\varone}{\ctxtwo}\isubs{\vartwo/\totunf{\tmtwo}}=_{\ih} \varone\isubs{\vartwo/\totunf{\tmtwo}}=\varone;\\
    \ref{l:p:parunf-ctx-one}.&&\parunf{\esub{\tmone}{\varone}{\tmthree}}{\esub{\ctxtwo}{\vartwo}{\tmtwo}}&=
      \parunf{\esub{\tmone}{\varone}{\tmthree}}{\ctxtwo}\isubs{\vartwo/\totunf{\tmtwo}}\\
       &&&=_{\ih}\parunf{\tmone}{\ctx{\ctxtwo}{\esub{\hole}{\varone}{\tmthree}}}\isubs{\vartwo/\totunf{\tmtwo}}=
      \parunf{\tmone}{\esub{\ctx{\ctxtwo}{\esub{\hole}{\varone}{\tmthree}}}{\vartwo}{\tmtwo}}.\\  
    \ref{l:p:parunf-ctx-two}.&&\parunf{(\l \varone.\tmone)}{\esub{\ctxtwo}{\vartwo}{\tmtwo}}&=
       \parunf{(\l \varone.\tmone)}{\ctxtwo}\isubs{\vartwo/\totunf{\tmtwo}} 
       =_{\ih} \l\varone.(\parunf{\tmone}{\ctx{\ctxtwo}{\l \varone.\hole}})\isubs{\vartwo/\totunf{\tmtwo}}\\
       &&&=\l\varone.(\parunf{\tmone}{\ctx{\ctxtwo}{\l \varone.\hole}}\isubs{\vartwo/\totunf{\tmtwo}})
       = \l\varone.(\parunf{\tmone}{\esub{\ctx{\ctxtwo}{\l \varone.\hole}}{\vartwo}{\tmtwo}}).\\
    \ref{l:p:parunf-ctx-three}.&&\parunf{(\tmone\ \tmthree)}{\esub{\ctxtwo}{\vartwo}{\tmtwo}} &=
       \parunf{(\tmone\ \tmthree)}{\ctxtwo} \isubs{\vartwo/\totunf{\tmtwo}}=_{\ih}
       (\parunf{\tmone}{\ctx{\ctxtwo}{\hole\ \tmthree}}\ \parunf{\tmthree}{\ctx{\ctxtwo}{\tmone\ \hole}})\isubs{\vartwo/\totunf{\tmtwo}}\\
       &&&=(\parunf{\tmone}{\ctx{\ctxtwo}{\hole\ \tmthree}}\isubs{\vartwo/\totunf{\tmtwo}})\ (\parunf{\tmthree}{\ctx{\ctxtwo}{\tmone\ \hole}}\isubs{\vartwo/\totunf{\tmtwo}})\\
       &&&=\parunf{\tmone}{\esub{\ctx{\ctxtwo}{\hole\ \tmthree}}{\vartwo}{\tmtwo}}\ \parunf{\tmthree}{\esub{\ctx{\ctxtwo}{\tmone\ \hole}}{\vartwo}{\tmtwo}};\\
    \ref{l:p:parunf-ctx-four}.&&\parunf{\esub{\tmone}{\varone}{\tmthree}}{\esub{\ctxtwo}{\vartwo}{\tmtwo}}&=
       \parunf{\esub{\tmone}{\varone}{\tmthree}}{\ctxtwo}\isubs{\vartwo/\totunf{\tmtwo}}\\
       &&&=_{\ih}\parunf{\tmone}{\ctx{\ctxtwo}{\esub{\hole}{\varone}{\tmthree}}}\isubs{\vartwo/\totunf{\tmtwo}}=
       \parunf{\tmone}{\esub{\ctx{\ctxtwo}{\esub{\hole}{\varone}{\tmthree}}}{\vartwo}{\tmtwo}}.\\
  \end{align*}
\end{varitemize}
This concludes the proof.
\end{proof}

\begin{lemma}
\label{l:free-var-rel-unf}
If $(\tmone,\ctxone)\unf{\valone}(\tmtwo,\ctxtwo)$ and $\valone\neq\bot$ then $\pi_1(\valone)=\fv{\parunf{\tmone}{\ctxone}}$ and $\pi_2(\valone)=\fv{\parunf{\tmtwo}{\ctxtwo}}$, where $\pi_1$ and $\pi_2$ are the usual projections functions.
\end{lemma}

\begin{proof}
By induction on the derivation $(\tmone,\ctxone)\unf{\valone}(\tmtwo,\ctxtwo)$. The rule $\rvar$ is the only atomic case, and the statement trivially holds. The inductive cases:
\begin{varitemize}
\item 
  $\runfl$) By \ih\ $\pi_1(\valone)=\fv{\parunf{\tmone}{\ctx{\ctxone}{\esub{\ctx{\ctxthree}{x}}{\varone}{ \hole}}}}$. We have to prove that 
  $\pi_1(\valone)=\fv{\parunf{\varone}{\ctx{\ctxone}{\esub{\ctxthree}{\varone}{\tmone}}}}$.  Lemma \ref{l:parunf-ctx}.\ref{l:p:parunf-ctx-four} 
  gives $\parunf{\tmone}{\ctx{\ctxone}{\esub{\ctx{\ctxthree}{x}}{\varone}{ \hole}}}=\parunf{\varone}{\ctx{\ctxone}{\esub{\ctxthree}{\varone}{\tmone}}}$, and so we conclude.
\item 
  $\runfr$) As the previous case.
\item 
  $\rsubl$) By \ih\ $\pi_1(\valone)=\fv{\parunf{\tmone}{\ctx{\ctxone}{\esub{ \hole}{\varone}{\tmthree}}}}$. We have to prove that 
  $\pi_1(\valone)=\fv{\parunf{\esub{\tmone}{\varone}{\tmthree}}{\ctxone}}$.  By lemma \ref{l:parunf-ctx}.\ref{l:p:parunf-ctx-one} 
  $\parunf{\tmone}{\ctx{\ctxone}{\esub{ \hole}{\varone}{\tmthree}}}=\parunf{\esub{\tmone}{\varone}{\tmthree}}{\ctxone}$, and  so we conclude.
\item 
  $\rsubr$) As the previous case.
\item 
  $\rabsoc$) By \ih\ $\pi_1(\valone)=\fv{\parunf{\tmone}{\ctx{\ctxone}{\l \varone.  \hole}}}$. We have to prove that 
  $\pi_1(\valone)\setminus\set{\varone}=\fv{\parunf{(\l \varone.\tmone)}{\ctxone}}$.  By Lemma 
  \ref{l:parunf-ctx}.\ref{l:p:parunf-ctx-two} $\l \varone.(\parunf{\tmone}{\ctx{\ctxone}{\l \varone.  \hole}})=\parunf{(\l \varone.\tmone)}{\ctxone}$, 
  and $\fv{\l \varone.(\parunf{\tmone}{\ctx{\ctxone}{\l \varone.  \hole}})}=\fv{\parunf{\tmone}{\ctx{\ctxone}{\l \varone.  \hole}}}\setminus\set{\varone}=\pi_1(\valone)\setminus\set{\varone}$. 
  Similarly for $(\tmtwo, \ctx{\ctxtwo}{\l \vartwo.  \hole})$.
\item 
  $\rabsweak$) By \ih\ $\pi_1(\valone)=\fv{\parunf{\tmone}{\ctx{\ctxone}{\l \varone.  \hole}}}$. The reasoning is as in the previous point, except that the 
  hypothesis $(\varone,\varthree)\notin \valone$ for all $\varthree$ implies that $\varone\notin\pi_1(\valone)$ and so 
  $\pi_1(\valone)\setminus\set{\varone}=\pi_1(\valone)$. Similarly for $(\tmtwo, \ctx{\ctxtwo}{\l \vartwo.  \hole})$.
\item $\rapp$) 
  By hypothesis $\valone\valop\valtwo\neq\bot$, which happens only if $\valone\neq\bot\neq\valtwo$ and $\valone\coh\valtwo$. In that case 
  $\valone\valop\valtwo=\valone\cup\valtwo$ and so $\pi_1(\valone\valop\valtwo)=\pi_1(\valone)\cup\pi_1(\valtwo)$.
  Since $\valone\neq\bot\neq\valtwo$ we can apply the \ih\ to the hypothesis of the rule, and get $\pi_1(\valone)=\fv{\parunf{\tmone}{\ctx{\ctxone}{ \hole\ \tmthree}}}$ and 
  $\pi_1(\valtwo)=\parunf{\tmthree}{\ctx{\ctxone}{\tmone\  \hole}}$. We have to prove that $\pi_1(\valone)\cup\pi_1(\valtwo)=\fv{\parunf{\tmone\ \tmthree}{\ctxone}}$.
  By Lemma \ref{l:parunf-ctx}.\ref{l:p:parunf-ctx-three} we get $\parunf{\tmone\ \tmthree}{\ctxone}=\parunf{\tmone}{\ctx{\ctxone}{ \hole\ \tmthree}}\ \parunf{\tmthree}{\ctx{\ctxone}{\tmone\  \hole}}$, 
  and so we conclude. Similarly for $(\tmtwo\ \tmfour, \ctxtwo)$.
\item 
  For other rules $\valone=\bot$, so there is nothing to prove.
\end{varitemize}
\end{proof}

\begin{definition}[Unifying Renaming]
A constraining set $\valone=\set{(\varone_1,\vartwo_1),\ldots,(\varone_k,\vartwo_k)}$ is a 
\emph{unifying renaming} for two  $\l$-terms $\tmone$ and $\tmtwo$ if $\valone$ is a bijection of $\fv{\tmone}$ 
and $\fv{\tmtwo}$ such that:
\begin{varenumerate}
\item 
  $\tmtwo\isubs{\varone_1/\vartwo_1}\ldots\isubs{\varone_k/\vartwo_k}=\tmone$.
\item 
  $\tmone\isubs{\vartwo_1/\varone_1}\ldots\isubs{\vartwo_k/\varone_k}=\tmtwo$.
\end{varenumerate}
\end{definition}

\begin{lemma}
\label{l:unif-ren-uniq}
If there is a unifying renaming between $\tmone$ and $\tmtwo$ then it is unique.
\end{lemma}
\begin{proof}
By induction on $\tmone$. Straightforward.
\end{proof}

\begin{lemma}
\label{l:auto-coherence}
If $(\tmone,\ctxone)\unf{\valone}(\tmtwo,\ctxtwo)$ and $\valone\neq\bot$ then $\valone$ is auto-coherent, \ie\ it is the graph of a bijection.
\end{lemma}
\begin{proof}
By induction on the derivation $(\tmone,\ctxone)\unf{\valone}(\tmtwo,\ctxtwo)$. For rule $\rvar$ it is obvious. For rules $\runfl$, $\runfr$, $\rsubl$, $\rsubr$, $\rabsoc$ and 
$\rabsweak$ it simply follows from the \ih. For rule $\rapp$ the \ih\ gives that $\valone$ and $\valtwo$ are two bijections. Now, $\valone\valop\valtwo\neq\bot$ implies that 
$\valone\neq\bot\neq\valtwo$ and $\valone\coh\valtwo$. Suppose that  $\valone\valop\valtwo=\valone\cup\valtwo$ is not a bijection. This happens iff $\valone\valop\valtwo$ 
is not a function or it is not injective, but since $\valone$ and $\valtwo$ are bijections both cases are absurd because of $\valone\coh\valtwo$.
\end{proof}

\begin{lemma}
\label{l:auto-coh-imp-unif-ren}
$(\tmone,\ctxone)\unf{\valone}(\tmtwo,\ctxtwo)$ and $\valone$ auto-coherent implies $\valone$ is a 
unifying renaming for $\parunf{\tmone}{\ctxone}$ and $\parunf{\tmtwo}{\ctxtwo}$.
\end{lemma}
\begin{proof}
By induction on the derivation of $(\tmone,\ctxone)\unf{\conone}(\tmtwo,\ctxtwo)$, using lemma \ref{l:parunf-ctx}. Cases:
\begin{varitemize}
\item 
  $\rvar$) Immediate.
\item 
  $\rsubl$) By \ih\ $\conone$ is a unifying renaming for $\parunf{\tmone}{\ctx{\ctxone}{\esub{\hole}{\varone}{\tmthree}}}$ and $\parunf{\tmtwo}{\ctxtwo}$. 
  By lemma \ref{l:p:parunf-ctx-one} we get:
  $$
  \parunf{\esub{\tmone}{\varone}{\tmthree}}{\ctxone}=\parunf{\tmone}{\ctx{\ctxone}{\esub{\hole}{\varone}{\tmthree}}}
  $$
  Thus $\conone$ is an unifying renaming  for $\parunf{\esub{\tmone}{\varone}{\tmthree}}{\ctxone}$ and $\parunf{\tmtwo}{\ctxtwo}$.
\item 
  $\rsubr$) As the previous one.
\item $\rabsoc$) 
  By \ih\ $\conone\cup(\varone,\vartwo)$ is a unifying renaming for $\parunf{\tmone}{\ctx{\ctxone}{\l \varone. \hole}}$ and $\parunf{\tmtwo}{\ctx{\ctxtwo}{\l \vartwo. \hole}}$ 
  and thus by definition $\varone\in\fv{\parunf{\tmone}{\ctx{\ctxone}{\l \varone. \hole}}}$ and $\vartwo\in\fv{\parunf{\tmtwo}{\ctx{\ctxtwo}{\l \vartwo. \hole}}}$. 
  By lemma \ref{l:parunf-ctx}.\ref{l:p:parunf-ctx-one} we get:
  $$
  \parunf{(\l \varone.\tmone)}{\ctxone}=\l \varone.(\parunf{\tmone}{\ctx{\ctxone}{\l \varone.\hole}})
  $$
  And similarly $\parunf{(\l \vartwo.\tmtwo)}{\ctxtwo} = \l \vartwo.(\parunf{\tmtwo}{\ctx{\ctxtwo}{\l \vartwo.\hole}})$. 
  Since $\fv{\l \varone.(\parunf{\tmone}{\ctx{\ctxone}{\l \varone.\hole}})}=\fv{\parunf{\tmone}{\ctx{\ctxone}{\l \varone.\hole}}}\setminus\set{\varone}$ 
  and $\fv{\l \vartwo.(\parunf{\tmtwo}{\ctx{\ctxtwo}{\l \vartwo.\hole}})}=\fv{\parunf{\tmtwo}{\ctx{\ctxtwo}{\l \vartwo.\hole}}}\setminus\set{\vartwo}$ we get that 
  $\conone$ is a unifying renaming for $\parunf{(\l \varone.\tmone)}{\ctxone}$ and $\parunf{(\l \vartwo.\tmtwo)}{\ctxtwo}$.
\item $\rabsweak$) By \ih\ $\conone$ is a unifying renaming for $\parunf{\tmone}{\ctx{\ctxone}{\l \varone. \hole}}$ and $\parunf{\tmtwo}{\ctx{\ctxtwo}{\l \vartwo. \hole}}$ s.t. $(\varone,\varthree), (\varthree,\vartwo)\notin\conone, \forall\varthree$. As in the previous case $\parunf{(\l \varone.\tmone)}{\ctxone} =\l \varone.(\parunf{\tmone}{\ctx{\ctxone}{\l \varone.\hole}})$
and $\parunf{(\l \vartwo.\tmtwo)}{\ctxtwo} = \l \vartwo.(\parunf{\tmtwo}{\ctx{\ctxtwo}{\l \vartwo.\hole}})$. By definition of unifying renaming we get that $\varone\notin\fv{\parunf{\tmone}{\ctx{\ctxone}{\l \varone. \hole}}}$ and $\vartwo\notin\fv{\parunf{\tmtwo}{\ctx{\ctxtwo}{\l \vartwo. \hole}}}$, and since $\fv{\parunf{\tmone}{\ctx{\ctxone}{\l \varone. \hole}}}=\fv{\l \varone. (\parunf{\tmone}{\ctx{\ctxone}{\l \varone. \hole}})}$ and $\fv{\parunf{\tmtwo}{\ctx{\ctxtwo}{\l \vartwo. \hole}}}=\fv{\l \vartwo.(\parunf{\tmtwo}{\ctx{\ctxtwo}{\l \vartwo. \hole}})}$ we get that $\conone$ is a unifying renaming of $\parunf{(\l \varone.\tmone)}{\ctxone}$ and $\parunf{(\l \vartwo.\tmtwo)}{\ctxtwo}$.

\item $\rapp$) By \ih\ $\conone$ is a unifying renaming for $\parunf{\tmone}{\ctx{\ctxone}{\hole\ \tmthree}}$ and $\parunf{\tmtwo}{ \ctx{\ctxtwo}{\hole\ \tmfour}}$, and $\contwo$ is a unifying renaming for $\parunf{\tmthree}{\ctx{\ctxone}{\tmone\ \hole}}$ and $\parunf{\tmfour}{\ctx{\ctxtwo}{\tmtwo\ \hole}}$. By lemma \ref{l:parunf-ctx}.\ref{l:p:parunf-ctx-three} we get 
$$\parunf{(\tmone\ \tmthree)}{\ctxone}=\parunf{\tmone}{\ctx{\ctxone}{\hole\ \tmthree}}\ \parunf{\tmthree}{\ctx{\ctxone}{\tmone\ \hole}}$$
and
$$\parunf{(\tmtwo\ \tmfour)}{\ctxtwo}=\parunf{\tmtwo}{\ctx{\ctxtwo}{\hole\ \tmfour}}\ \parunf{\tmfour}{\ctx{\ctxtwo}{\tmtwo\ \hole}}$$
Since $\fv{\parunf{(\tmone\ \tmthree)}{\ctxone}}=\fv{\parunf{\tmone}{\ctx{\ctxone}{\hole\ \tmthree}}}\cup\fv{\parunf{\tmthree}{\ctx{\ctxone}{\tmone\ \hole}}}$ (and analogously for $\parunf{(\tmtwo\ \tmfour)}{\ctxtwo}$) and $\conone\coh\contwo$, we get that $\conone\cup\contwo$ is a unifying renaming for $\parunf{(\tmone\ \tmthree)}{\ctxone}$ and $\parunf{(\tmtwo\ \tmfour)}{\ctxtwo}$.

\item $\runfl$) By \ih\ $\conone$ is a unifying renaming for $\parunf{\tmone}{\ctx{\ctxone}{\esub{\ctx{\ctxthree}{x}}{\varone}{\hole}}}$ and $\parunf{\tmtwo}{\ctxtwo}$. By lemma \ref{l:parunf-ctx}.\ref{l:p:parunf-ctx-four} we get 

$$\parunf{\tmone}{\ctx{\ctxone}{\esub{\ctx{\ctxthree}{x}}{\varone}{\hole}}}= \parunf{\varone}{\ctx{\ctxone}{\esub{\ctxthree}{\varone}{\tmone}}}$$
And so $\conone$ is a unifying renaming for $\parunf{\varone}{\ctx{\ctxone}{\esub{\ctxthree}{\varone}{\tmone}}}$ and $\parunf{\tmtwo}{\ctxtwo}$.

\item $\runfr$) As in the previous case.
\end{varitemize}
This concludes the proof.
\end{proof}

\begin{lemma}
\label{l:bot-imp-no-un-ren}
$(\tmone,\ctxone)\unf{\bot}(\tmtwo,\ctxtwo)$ implies there is no  
unifying renaming for $\parunf{\tmone}{\ctxone}$ and $\parunf{\tmtwo}{\ctxtwo}$.
\end{lemma}

\ben{
\begin{proof}
The relative unfolding of a lambda is always a lambda, and the relative unfolding of an application is always an application. Therefore for rules $\rerrla$ and $\rerral$ the relative unfoldings have different topmost constructors and thus are different. For rules $\rerrlv$, $\rerrvl$, $\rerrva$ and $\rerrav$ note that Lemma \ref{l:parunf-ctx}.\ref{l:p:parunf-ctx-zero} states that the relative unfolding of a non-substituted variable is a variable, thus again the relative unfoldings of the two terms have different topmost constructors and are different. For rule $\rabsinc$ if $\valone=\bot$ then it follows from the \ih\ and Lemma \ref{l:parunf-ctx}.\ref{l:p:parunf-ctx-two}. 

Otherwise suppose that there exists $\varthree\neq\vartwo$ s.t. $\{(\varone,\varthree)\}\in\valone$. By Lemmas \ref{l:auto-coherence} and \ref{l:auto-coh-imp-unif-ren} we get that $\valone$ is an unifying renaming for $\parunf{\tmone}{\ctx{\ctxone}{\l \varone.  \hole}}$ and $\parunf{\tmtwo}{\ctx{\ctxtwo}{\l \vartwo.  \hole}}$. Moreover, by Lemma \ref{l:free-var-rel-unf} we get $\pi_1(\valone)=\fv{\parunf{\tmone}{\ctx{\ctxone}{\l \varone.  \hole}}}$ and $\pi_2(\valone)=\fv{\parunf{\tmtwo}{\ctx{\ctxtwo}{\l \vartwo.  \hole}}}$. 

Now, by Lemma \ref{l:parunf-ctx}.\ref{l:p:parunf-ctx-two} we get $\parunf{(\l \varone.\tmone)}{\ctxone} = \l\varone.(\parunf{\tmone}{\ctx{\ctxone}{\l \varone.\hole}})$ and $\parunf{(\l \vartwo.\tmtwo)}{\ctxtwo} = \l\vartwo.(\parunf{\tmtwo}{\ctx{\ctxtwo}{\l \vartwo.\hole}})$. From $\{(\varone,\varthree)\}\in\valone$ we get that $\varone\in\fv{\parunf{\tmone}{\ctx{\ctxone}{\l \varone.  \hole}}}$. It is easily seen that in this case if there is an unifying renaming for $\parunf{(\l \varone.\tmone)}{\ctxone}$ and $\parunf{(\l \vartwo.\tmtwo)}{\ctxtwo}$ it is necessarily obtained from an unifying renaming for $\parunf{\tmtwo}{\ctx{\ctxtwo}{\l \vartwo.  \hole}}$ and $\parunf{\tmtwo}{\ctx{\ctxtwo}{\l \vartwo.  \hole}}$ by removing the pair $(\varone,\vartwo)$. We know that $\valone$ is an unifying renaming for $\parunf{\tmtwo}{\ctx{\ctxtwo}{\l \vartwo.  \hole}}$ and $\parunf{\tmtwo}{\ctx{\ctxtwo}{\l \vartwo.  \hole}}$. By Lemma \ref{l:auto-coherence} it is a bijection. By hypothesis it contains the pair $(\varone,\varthree)$, and so it cannot contain the pair $(\varone,\vartwo)$. So there cannot be an unifying renaming for $\parunf{(\l \varone.\tmone)}{\ctxone}$ and $\parunf{(\l \vartwo.\tmtwo)}{\ctxtwo}$.
\end{proof}
}}{}

\begin{theorem}[Correctness]
The Unfolding Checking Algorithm, on input $(\btmone,\btmtwo)$,
returns $\yes$ iff $\totunf{\btmone}=\totunf{\btmtwo}$.
\end{theorem}
\condinc{}{
\begin{proof}
This can be done by proving the following two lemmas:
\begin{varitemize}
\item
  If $(\tmone,\ctxone)\unf{\valone}(\tmtwo,\ctxtwo)$ and 
  $\valone\neq \bot$ then $\valone$ induces two renamings $\sigma_1$ and $\sigma_2$ such that 
  $\parunf{\tmone}{\ctxone}\sigma_1=\parunf{\tmtwo}{\ctxtwo}$ and $\parunf{\tmtwo}{\ctxtwo}\sigma_2=\parunf{\tmone}{\ctxone}$;
\item
  If $(\tmone,\ctxone)\unf{\bot}(\tmtwo,\ctxtwo)$, then there is no  
  such pair of renamings for $\parunf{\tmone}{\ctxone}$ and $\parunf{\tmtwo}{\ctxtwo}$.
\end{varitemize}
Both can be proved by induction on the derivation of 
$(\tmone,\ctxone)\unf{\conone}(\tmtwo,\ctxtwo)$. For details, see~\cite{ExtendedVersion}
\end{proof}}
This is not the end of the story, however --- one also needs to be sure about the complexity of the algorithm, which
turns out to be polynomial:
\condinc{}{}
\begin{proposition}[Complexity]
The Unfolding Checking Algorithm works in time polynomial in $|\btmone|+|\btmtwo|$.
\end{proposition}
\begin{proof}
The following observations are sufficient to obtain the thesis:
\begin{varitemize}
\item
  The number of entries in $\matr$ is $|\btmone||\btmtwo|$ in total.
\item
  At every iteration, one element of $\matr$ changes its value from $\blank$ to
  some non-blank $\valone$.
\item
  Step \ref{s:iter} can clearly be performed in time polynomial in $|\btmone|+|\btmtwo|$.
\item
  Computing the predecessors of a pair $\pone$ can be done in
  polynomial time, and so Step \ref{s:choose} can itself be performed
  in time polynomial in $|\btmone|+|\btmtwo|$.
\item
  Rules in Figure~\ref{f:unfolding-rules} can all be applied in polynomial time,
  in particular \condinc{due to Lemma~\ref{lemma:cardval}}{because the cardinality of any constraining set produced
  by the algorithm can be $|\btmone|+|\btmtwo|$, at most}. As a consequence,
  Step \ref{s:compute} can be performed in polynomial time.
\end{varitemize}
\end{proof}

\section{Encoding Turing Machines}\label{s:tur-mach}
A cost model \emph{for computation time} is said to be invariant if it is polynomially related to the standard cost model
on Turing machines. In sections~\ref{s:expsubst} and \ref{s:relation}, we
proved that head reduction of any $\lambda$-term $\tmone$ can be performed on a Turing machine in time
polynomial in the number of $\beta$-steps leading $\tmone$ to its normal form (provided it exists). This
is proved through explicit substitutions, which in Section~\ref{s:unfolding} are shown to be
a reasonable representation for $\lambda$-terms: two terms $\tmone$ and $\tmtwo$ in $\lm$
can be checked to have the same unfolding in polynomial time.

The last side of the story is still missing, however. In this section, we will show
how head reduction can simulate Turing machine computation in such a way that derivational
complexity of the simulating $\lambda$-term is polynomially related to the running time of the encoded
Turing machine. Results similar to the one we are going to present are common for the 
$\lambda$-calculus, so we will not give all the details, which can anyway be found 
in~\cite{ExtendedVersion}.

\newcommand{\fix}{H}
\condinc{
The first thing we need to encode is a form of recursion. We denote by $\fix$ the term $\tmone\tmone$, where
$
\tmone\equiv\lambda\varone.\lambda\vartwo.\vartwo(\varone\varone\vartwo).
$
$\fix$ is a call-by-name fixed-point operator: for every term $\tmtwo$, 
$$
\fix\tmtwo\toh(\lambda\vartwo.\vartwo(\tmone\tmone\vartwo))\tmtwo\toh\tmtwo(\fix\tmtwo).
$$
The $\lambda$-term $\fix$ provides the necessary computational expressive power
to encode the whole class of computable functions.}{}

\newcommand{\finsetone}{A}
\newcommand{\elone}{a}
\newcommand{\cod}[2]{\lceil #1\rceil^{#2}}
\newcommand{\alpone}{\Sigma}
\newcommand{\alptwo}{\Delta}
\newcommand{\strone}{s}
\newcommand{\strtwo}{r}
\newcommand{\appendchar}[1]{\mathit{AC}(#1)}
The simplest objects we need to encode in the $\lambda$-calculus are finite
sets. Elements of any finite set $\finsetone=\{\elone_1,\ldots,\elone_n\}$ can be encoded
as follows:
$
\cod{\elone_i}{\finsetone}\equiv\lambda\varone_1.\ldots.\lambda\varone_n.\varone_i \;.
$
Notice that the above encoding induces a total order
on $A$ such that $a_i\leq a_j$ iff $i\leq j$. Other useful objects are finite strings over an arbitrary alphabet,
which will be encoded using a scheme attributed to Scott.
Let $\alpone=\{\elone_1,\ldots,\elone_n\}$ be a finite alphabet. A string in
$\strone\in\alpone^*$ can be represented by a value
$\cod{\strone}{\alpone^*}$ as follows, by induction on the structure of $\strone$:
$$\begin{array}{cccccccccccccccc}
\cod{\varepsilon}{\alpone^*}&\equiv&\lambda\varone_1.\ldots.\lambda\varone_n.\lambda\vartwo.\vartwo\; &&&&&&&
\cod{\elone_i\strtwo}{\alpone^*}&\equiv&\lambda\varone_1.\ldots.
   \lambda\varone_n.\lambda\vartwo.\varone_i\cod{\strtwo}{\alpone^*}
\end{array}$$
Observe that representations of symbols in $\alpone$
and strings in $\alpone^*$ depend on the cardinality
of $\alpone$. In other words, if $\strone\in\alpone^*$
and $\alpone\subset\alptwo$, $\cod{\strone}{\alpone^*}\neq\cod{\strone}{\alptwo^*}$.

\newcommand{\convertstring}[2]{\mathit{CS}(#1,#2)}
Of course, one should be able to very easily compute the encoding of a string obtained
by concatenating another string with a character. Moreover, the way strings are
encoded depends on the underlying alphabet, and as a consequence,
we also need to be able to convert representations for strings in one alphabet to
corresponding representations in another, different, alphabet. This can be done efficiently
in the $\lambda$-calculus by way of a term $\appendchar{\alpone}$ which append a
character to a string (both expressed in the alphabet $\alpone$) and a term
$\convertstring{\alpone}{\alptwo}$ which converts a string $\strone\in\alpone^*$ 
into another string in $\alptwo$ obtained by replacing any character in $\alpone-\alptwo$
by the empty string.  $\appendchar{\alpone}$ works in time independent on the
size of the input, while $\convertstring{\alpone}{\alptwo}$ works in time \emph{proportional}
to the size of the argument.
\condinc{
\newcommand{\constAC}[1]{n_{\mathit{AC}}^{#1}}
\begin{lemma}
Given a finite alphabet $\alpone$, there are
a term $\appendchar{\alpone}$ and a constant $\constAC{\alpone}\in\Nset$ 
such that for every $\elone\in\alpone$, every
term $\tmone$ and every $\strone\in\alpone^*$, there is $n\leq\constAC{\alpone}$ such that
$\appendchar{\alpone}\tmone\cod{a}{\alpone}\cod{u}{\alpone^*}\toh^n\tmone\cod{au}{\alpone^*}$.
\end{lemma}
\begin{proof}
The term we are looking for is defined as follows:
$$
\appendchar{\alpone}\equiv\lambda y.\lambda a.\lambda u.aM_1\ldots M_{|\alpone|}uy
$$
where for any $i$, $M_i \equiv \lambda u.\lambda y.y(\lambda x_1.\ldots.\lambda x_{|\alpone|}.\lambda w.x_iu)$.
Observe that:
\begin{eqnarray*}
\appendchar{\alpone}\tmone\cod{a_i}{\alpone}\cod{u}{\alpone^*}&\toh^3&
  \cod{a_i}{\alpone}M_1\ldots M_{|\alpone|}\cod{u}{\alpone^*}t\\
  &\toh^{|\alpone|}&M_i\cod{u}{\alpone^*}t\toh^2\tmone\cod{a_iu}{\alpone^*}.
\end{eqnarray*}
In other words, $\constAC{\alpone}$ can be set to be $|\alpone|+5$. This concludes the proof.
\end{proof}
\newcommand{\forget}[2]{\mathcal{G}^{#1}_{#2}}
Given alphabets $\alpone$ and $\alptwo$, the
function $\forget{\alpone}{\alptwo}:\alpone^*\rightarrow\alptwo^*$
is defined by stipulating that for every $n\in\Nset$ and every
$a_1,\ldots,a_n\in\alpone$, $\forget{\alpone}{\alptwo}(a_1\ldots a_n)=u_1\ldots u_n$.
where $u_i$ is $a_i$ if $a_i\in\alptwo$ and $u_i$ is $\varepsilon$ otherwise
\newcommand{\functionCS}[2]{f_{\mathit{CS}}^{#1,#2}}
\begin{lemma}
Given finite alphabets $\alpone$ and $\alptwo$, 
there are a term $\convertstring{\alpone}{\alptwo}$ and a linear
function $\functionCS{\alpone}{\alptwo}:\Nset\rightarrow\Nset$ such that for every $u\in\alpone$
there is $n\leq\functionCS{\alpone}{\alptwo}(|u|)$ such that 
$\convertstring{\alpone}{\alptwo}\tmone\cod{u}{\alpone^*}\toh^n\tmone\cod{\forget{\alpone}{\alptwo}(u)}{\alptwo^*}$.
\end{lemma}
\begin{proof}
The term we are looking for is defined as follows:
$$
\convertstring{\alpone}{\alptwo}\equiv H(\lambda x.\lambda z.\lambda u.uN_1\ldots N_{|\alpone|}Nz),
$$
where for any $i$,
$$
N_i\equiv\left\{
  \begin{array}{ll}
    \lambda u.\lambda z.x(\lambda u.\appendchar{\alpone}z\cod{a_i}{\alptwo}u)u & \mbox{if $a_i\in\alptwo$}\\
    \lambda u.\lambda z.x(\lambda u.zu)u &\mbox{otherwise.}
  \end{array}
\right.
$$
and $N\equiv\lambda z.z\cod{\varepsilon}{\alptwo^*}$.
Let $P_i$ be $N_i\{\convertstring{\alpone}{\alptwo}/x\}$. Then the thesis
can be proved by induction on $u$ as soon as $\functionCS{\alpone}{\alptwo}$ is defined
as $\functionCS{\alpone}{\alptwo}(x)=x(|\alpone|+\constAC{\alpone}+7)+|\alpone|+5$:
\begin{eqnarray*}
 \convertstring{\alpone}{\alptwo}\tmone\cod{\varepsilon}{\alpone^*}&\toh^2&
   (\lambda z.\lambda u.uP_1\ldots P_{|\alpone|}Nz)\tmone\cod{\varepsilon}{\alpone^*}\\
   &\toh^2& \cod{\varepsilon}{\alpone^*}P_1\ldots P_{|\alpone|}N\tmone\\
   &\toh^{|\alpone|+1}& N\tmone\toh\tmone\cod{\varepsilon}{\alptwo^*}\\\
 \convertstring{\alpone}{\alptwo}\cod{a_iu}{\alpone^*}&\toh^2&
   (\lambda z.\lambda u.uP_1\ldots P_{|\alpone|}Nz)\tmone\cod{a_iu}{\alpone^*}\\
   &\toh^2& \cod{a_iu}{\alpone^*}P_1\ldots P_{|\alpone|}N\tmone\\
   &\toh^{|\alpone|+1}& P_i\cod{u}{\alpone^*}\tmone
\end{eqnarray*}
Now, if $a_i\in\alptwo$, then there are $n\leq\functionCS{\alpone}{\alptwo}(|u|)$ and $m\leq\constAC{\alpone}$ such that
\begin{eqnarray*}
  P_i\cod{u}{\alpone^*}\tmone&\toh^2&\convertstring{\alpone}{\alptwo}(\lambda u.\appendchar{\alptwo}\tmone\cod{a_i}{\alptwo}u)\cod{u}{\alpone^*}\\
        &\toh^n&(\lambda u.\appendchar{\alptwo}\tmone\cod{a_i}{\alptwo}u)\cod{\forget{\alpone}{\alptwo}(u)}{\alptwo^*}\\
        &\toh&\appendchar{\alptwo}\tmone\cod{a_i}{\alptwo}\cod{\forget{\alpone}{\alptwo}(u)}{\alptwo^*}\\
        &\toh^m&\tmone\cod{a_i\forget{\alpone}{\alptwo}(u)}{\alptwo^*}
\end{eqnarray*}
while otherwise there is $n\leq\functionCS{\alpone}{\alptwo}(|u|)$ such that
\begin{eqnarray*}
  P_i\cod{u}{\alpone^*}\tmone&\toh^2&\convertstring{\alpone}{\alptwo}(\lambda u.\tmone u)\cod{u}{\alpone^*}\\
        &\toh^n&(\lambda u.\tmone u)\cod{\forget{\alpone}{\alptwo}(u)}{\alptwo^*}\\
        &\toh&\tmone\cod{\forget{\alpone}{\alptwo}(u)}{\alptwo^*}
\end{eqnarray*}
This concludes the proof.
\end{proof}
A deterministic Turing machine $\mathcal{M}$ is a tuple $(\alpone,a_{\mathit{blank}},Q,q_{\mathit{initial}},
q_{\mathit{final}},\delta)$ consisting of:
\begin{varitemize}
  \item
  A finite alphabet $\alpone=\{a_1,\ldots,a_n\}$;
  \item
  A distinguished symbol $a_{\mathit{blank}}\in\alpone$, called the \emph{blank symbol};
  \item
  A finite set $Q=\{q_1,\ldots,q_m\}$ of \emph{states};
  \item
  A distinguished state $q_{\mathit{initial}}\in Q$, called the \emph{initial
  state};
  \item
  A distinguished state $q_{\mathit{final}}\in Q$, called the \emph{final
  state}; 
  \item
  A partial \emph{transition function} $\delta:Q\times\alpone\rightharpoonup 
  Q\times\alpone\times\{\leftarrow,\rightarrow,\downarrow\}$ such that
  $\delta(q_i,a_j)$ is defined iff $q_i\neq q_{\mathit{final}}$.
\end{varitemize}
A configuration for $\mathcal{M}$ is a quadruple in
$\alpone^*\times\alpone\times\alpone^*\times Q$. 
For example, if $\delta(q_i,a_j)=(q_l,a_k,\leftarrow)$, then $\mathcal{M}$ 
evolves from $(ua_p,a_j,v,q_i)$ to $(u,a_p,a_kv,q_l)$ (and from $(\varepsilon,a_j,v,q_i)$
to $(\varepsilon,a_{\mathit{blank}},a_kv,q_l)$). A configuration like
$(u,a_i,v,q_{\mathit{final}})$ is \emph{final} and cannot evolve.
Given a string $u\in\alpone^*$ which does not contain any occurrence
of $a_{\mathit{blank}}$, the \emph{initial configuration}
for $u$ is $(\varepsilon,a_{\mathit{blank}},u,q_{\mathit{initial}})$,
while the \emph{final} for $u$ is $(\varepsilon,a_{\mathit{blank}},u,q_{\mathit{final}})$.

A Turing machine $(\alpone,a_{\mathit{blank}},Q,q_{\mathit{initial}},
q_{\mathit{final}},\delta)$ computes the function $f:\alptwo^*\rightarrow\alptwo^*$ 
(where $\alptwo\subseteq\alpone$ and $a_{\mathit{blank}}$ is not in $\alptwo$) in time $g:\Nset\rightarrow\Nset$ 
iff for every $u\in\alptwo^*$, the
initial configuration for $u$ evolves to a final configuration
for $f(u)$ in $g(|u|)$ steps.

A configuration $(s,a,v,q)$ of a Turing machine
$\M=(\alpone,a_{\mathit{blank}},Q,q_{\mathit{initial}},
q_{\mathit{final}},\delta)$ is represented by
the term $\cod{(u,a,v,q)}{\M}\equiv\lambda x.x\cod{u^r}{\alpone^*}
\;\cod{a}{\alpone}\;\cod{v}{\alpone^*}\;\cod{q}{Q}$.}
{

A configuration $(u,a,v,q)$ of a Turing machine
$\M=(\alpone,a_{\mathit{blank}},Q,q_{\mathit{initial}},
q_{\mathit{final}},\delta)$ consists of the portion of the
tape on the left of the head (namely $u\in\alpone^*$), the symbol
under the head ($a\in\alpone$), the tape on the right of the head
($v\in\alpone^*$), and the current state $q\in Q$. 
Such a configuration is represented by
the term $\cod{(u,a,v,q)}{\M}\equiv\lambda x.x\cod{u^r}{\alpone^*}
\;\cod{a}{\alpone}\;\cod{v}{\alpone^*}\;\cod{q}{Q}$.
}

We now encode a Turing machine 
$\M=(\alpone,a_{\mathit{blank}},Q,q_{\mathit{initial}},
q_{\mathit{final}},\delta)$ in the $\lambda$-calculus.
Suppose $\alpone=\{a_1,\ldots,a_{|\alpone|}\}$
and $Q=\{q_1,\ldots,q_{|Q|}\}$.
The encoding of $\M$ is defined around three $\lambda$-terms:
\newcommand{\init}[2]{\mathcal{I}(#1,#2)}
\newcommand{\functioninit}[2]{f_{\mathcal{I}}^{#1,#2}}
\newcommand{\final}[2]{\mathcal{F}(#1,#2)}
\newcommand{\functionfinal}[2]{f_{\mathcal{F}}^{#1,#2}}
\newcommand{\trans}[1]{\mathcal{T}(#1)}
\newcommand{\functiontrans}[1]{f_{\mathcal{T}}^{#1}}
\begin{varitemize} 
  \item
    First of all, we need to be able to build the initial
    configuration for $u$ from $u$ itself. This can be done in 
    time proportional to $|u|$ by a term $\init{\M}{\alptwo}$, where
    $\alptwo$ is the alphabet of $u$, which can be different from $\alpone$.
    $\init{\M}{\alptwo}$ simply converts $u$ into the
    appropriate format by way of $\convertstring{\alptwo}{\alpone}$,
    and then packages it into a configuration.
  \item
    Then, we need to extract a string from a final configuration
    $C$ for the string. This can be done in time proportional to
    the size of $C$ by a term $\final{\M}{\alptwo}$, which makes
    essential use of $\convertstring{\alpone}{\alptwo}$.
  \item
    Most importantly, we need to be able to simulate the transition function of
    $\M$, i.e. compute a final configuration
    from an initial configuration (if it exists). This
    can be done with cost proportional to the number of
    steps $\M$ takes on the input, by way of a term $\trans{\M}$. 
    The term $\trans{\M}$ just performs case analysis depending
    on the four components of the input configuration, then
    manipulating them making use of $\appendchar{\alpone}$.
\end{varitemize}
\condinc{
The following three lemmas formalize the above intuitive
argument:
\begin{lemma}
Given a Turing machine $\M=(\alpone,a_{\mathit{blank}},Q,q_{\mathit{initial}},
q_{\mathit{final}},\delta)$ and a blank-free alphabet $\alptwo$, 
there are a term $\init{\M}{\alptwo}$ and a linear
function $\functioninit{\M}{\alptwo}$ such that for every $u\in\alptwo^*$, 
there in $n\in\Nset$ such that
$\init{\M}{\alptwo}\tmone\cod{u}{\alptwo^*}\toh^n\tmone\cod{C}{\M}$
where $C$ is the initial configuration for $u$ and $n\leq\functioninit{\M}{\alptwo}(|u|)$.
\end{lemma}
\begin{proof}
Simply,
\begin{align*}
\init{\M}{\alptwo}&=\lambda x.\lambda u.\convertstring{\alptwo}{\alpone}(\lambda z.x(\lambda y.y\cod{\varepsilon}{\alpone^*}\;\cod{a_{\mathit{blank}}}{\alpone}\;z\;\cod{q}{Q}))u;\\
\functioninit{\M}{\alptwo}(x)&=\functionCS{\alpone}{\alptwo}(x)+3.
\end{align*}
Indeed,
\begin{align*}
\init{\M}{\alptwo}\tmone\cod{u}{\alptwo^*}&\toh^2\convertstring{\alptwo}{\alpone}(\lambda z.\tmone(\lambda y.y\cod{\varepsilon}{\alpone^*}\;\cod{a_{\mathit{blank}}}{\alpone}\;z\;\cod{q}{Q}))\cod{u}{\alptwo^*}\\
  &\toh^n(\lambda z.\tmone(\lambda y.y\cod{\varepsilon}{\alpone^*}\;\cod{a_{\mathit{blank}}}{\alpone}\;z\;\cod{q}{Q}))\cod{u}{\alpone^*}\\
  &\toh\tmone(\lambda y.y\cod{\varepsilon}{\alpone^*}\;\cod{a_{\mathit{blank}}}{\alpone}\;\cod{u}{\alpone^*}\;\cod{q}{Q}))\cod{u}{\alpone^*}
\end{align*}
where $n\leq\functionCS{\alptwo}{\alpone}{|u|}$. This concludes the proof.
\end{proof}
\begin{lemma}
Given a Turing machine $\M=(\alpone,a_{\mathit{blank}},Q,q_{\mathit{initial}},
q_{\mathit{final}},\delta)$ and for every alphabet $\alptwo$, there are
a term $\final{\M}{\alptwo}$ and a linear function
$\functionfinal{\M}{\alptwo}$ such that for every final
configuration $C$ for $u\in\alptwo^*$ there is $n\in\Nset$ 
$\final{\M}{\alptwo}\tmone\cod{C}{\M}\toh^n\tmone\cod{u}{\alptwo^*}$,
where $n\leq\functionfinal{\M}{\alptwo}(|u|)$. 
\end{lemma}
\begin{proof}
Simply,
\begin{align*}
\final{\M}{\alptwo}&=\lambda x.\lambda y.y(\lambda v.\lambda a.\lambda u.\lambda q.\convertstring{\alpone}{\alptwo}xu)\\
\functionfinal{\M}{\alptwo}(x)&=\functionCS{\alpone}{\alptwo}(x)+7
\end{align*}
Indeed,
\begin{align*}
\final{\M}{\alptwo}\tmone\cod{C}{\M}&\toh^2\cod{C}{\M}(\lambda v.\lambda a.\lambda u.\lambda q.\convertstring{\alpone}{\alptwo}\tmone u)
   &\toh^5\convertstring{\alpone}{\alptwo}\tmone\cod{u}{\alpone^*}\\
   &\toh^n\tmone\cod{u}{\alptwo^*}
\end{align*}
where $n\leq\functionCS{\alpone}{\alptwo}{|u|}$. This concludes the proof.
\end{proof}
\begin{lemma}
Given a Turing machine $\M=(\alpone,a_{\mathit{blank}},Q,q_{\mathit{initial}},
q_{\mathit{final}},\delta)$, there are a term $\trans{\M}$ and
a linear function $\functiontrans{\M}$ such that 
for every configuration $C$,
\begin{varitemize}
  \item
  if $D$ is a final configuration reachable from $C$ in $n$ steps,
  then $\trans{\M}\tmone\cod{C}{\M}\toh^m\tmone\cod{D}{\M}$
  where $m\leq\functiontrans{\M}(n)$;
  \item
  the term $\trans{\M}\tmone\cod{C}{\M}$ diverges if there is
  no final configuration reachable from $C$.
\end{varitemize}
\end{lemma}
\begin{proof}
$\trans{\M}$ is defined as
$$
H(\lambda x.\lambda z.\lambda y.y(\lambda u.\lambda a.\lambda v.\lambda q.q(M_1\ldots M_{|Q|})uavz)),
$$
where, for any $i$ and $j$:
\begin{eqnarray*}
M_i&\equiv&\lambda u.\lambda a.\lambda v.\lambda z.a(N_i^1\ldots N_i^{|\alpone|})uvz;\\
N_i^j&\equiv&
    \left\{
   \begin{array}{ll}
   \lambda u.\lambda v.\lambda z.z(\lambda x.xu\cod{a_j}{\alpone}v\cod{q_i}{Q})
      \hfill\mbox{if $q_i=q_\mathit{final}$}\\
   \lambda u.\lambda v.\lambda z.xz(\lambda z.zu\cod{a_k}{\alpone}v\cod{q_l}{Q}) 
      \hfill\mbox{\qquad if $\delta(q_i,a_j)=(q_l,a_k,\downarrow)$}\\
   \lambda u.\lambda v.\lambda z.uP_1^{l,k}\ldots P_{|\alpone|}^{l,k}P^{l,k}vz
      \hfill\mbox{if $\delta(q_i,a_j)=(q_l,a_k,\leftarrow)$}\\ 
   \lambda u.\lambda v.\lambda z.vR_1^{l,k}\ldots R_{|\alpone|}^{l,k}R^{l,k}uz
      \hfill \mbox{if $\delta(q_i,a_j)=(q_l,a_k,\rightarrow)$;}
   \end{array}
    \right.\\
P_i^{l,k}&\equiv&\lambda u.\lambda v.\lambda z.\appendchar{\alpone}(\lambda w.xz(\lambda x.xu\cod{a_i}{\alpone}w\cod{q_l}{Q}))\cod{a_k}{\alpone}v\\
P^{l,k}&\equiv&\lambda v.\lambda z.\appendchar{\alpone}(\lambda w.xz(\lambda x.x\cod{\varepsilon}{\alpone}\cod{a_{\mathit{blank}}}{\alpone}w\cod{q_l}{Q}))\cod{a_k}{\alpone}v\\
R_i^{l,k}&\equiv&\lambda v.\lambda u.\lambda z.\appendchar{\alpone}(\lambda w.xz(\lambda x.xw\cod{a_i}{\alpone}v\cod{q_l}{Q}))\cod{a_k}{\alpone}u\\
R^{l,k}&\equiv&\lambda u.\lambda z.\appendchar{\alpone}(\lambda w.xz(\lambda x.xw\cod{a_{\mathit{blank}}}{\alpone}\cod{\varepsilon}{\alpone^*}\cod{q_l}{Q}))\cod{a_k}{\alpone}u
\end{eqnarray*}
It is routine to prove the thesis.
 \end{proof}}
{
A peculiarity of the proposed encoding of Turing machines is the use of 
Scott numerals (instead of Church numerals), which make the simulation to work even when
head reduction is the underlying notion of computation. Another crucial aspect is \emph{continuation
passing}: the $\lambda$-terms cited above all take an additional parameter to which the result is
applied after being computed.
}
At this point, we can give the desired simulation result:
\newcommand{\function}[2]{\mathcal{U}(#1,#2)}
\begin{theorem}[Invariance, Part II]
If $f:\alptwo^*\rightarrow\alptwo^*$ is computed by a Turing machine
$\M$ in time $g$, then there is a term $\function{\M}{\alptwo}$
such that for every $u\in\alptwo^*$,
$\function{\M}{\alptwo}\cod{u}{\alptwo^*}\toh^n\cod{f(u)}{\alptwo^*}$
where $n=\mathcal{O}(g(|u|)+|u|)$.
\end{theorem}
\begin{proof}
Simply define $\function{\M}{\alptwo}\equiv \lambda u.\init{\M}{\alptwo}(\lambda x.\trans{\M}(\lambda y.\final{\M}{\alptwo}(\lambda w.w)y))y)u$.
It is routine to prove the thesis.
\end{proof}
Noticeably, the just described simulation induces a linear overhead: every step of $\M$ corresponds
to a constant cost in the simulation, the constant cost not depending on the input but only on
$\M$ itself.

\section{Conclusions}
The main result of this paper is the first invariance result for the $\lambda$-calculus when reduction
is allowed to take place in the scope of abstractions. The key tool to achieve invariance are linear explicit
substitutions, which are \emph{compact} but \emph{manageable} representations of $\lambda$-terms.

Of course, the main open problem in the area, namely invariance of the unitary cost model for any
normalizing strategy (e.g. for the strategy which always reduces the leftmost-outermost redex) remains
open. Although linear explicit substitutions cannot be \emph{directly} applied to this problem, 
the authors strongly believe that this is anyway a promising direction, on which they are actively
working at the time of writing. 
\bibliographystyle{plain}
\bibliography{main}

\end{document}